\documentclass[sigconf,nonacm,pbalance]{acmart}

\newif\ifdraft
\draftfalse

\AtBeginDocument{%
  }

\setcopyright{acmlicensed}
\copyrightyear{2018}
\acmYear{2018}
\acmDOI{XXXXXXX.XXXXXXX}

\acmConference[CCS '25]{Make sure to enter the correct
  conference title from your rights confirmation email}{October 13--17,
  2025}{Taipei, Taiwan}
\acmISBN{978-1-4503-XXXX-X/18/06}

\usepackage{xurl}
\usepackage{soul}

\usepackage{ifpdf}

\ifpdf
\else
  \usepackage{breakurl}
\fi

\usepackage{placeins}
\usepackage{amsmath}
\usepackage{amsthm}
\usepackage{nicefrac}
\newtheorem{theorem}{Theorem}
\newtheorem{lemma}[theorem]{Lemma}
\newtheorem{definition}[theorem]{Definition}
\newtheorem{note}[theorem]{Note}

\newtheorem{proposition}[theorem]{Proposition}

\usepackage{enumitem}
\usepackage{color}
\usepackage{framed}
\usepackage{graphicx}
\usepackage{hyperref} 
\hypersetup{
  colorlinks   = true,    % Colours links instead of ugly boxes
  urlcolor     = blue,    % Colour for external hyperlinks
  linkcolor    = blue,    % Colour of internal links
  citecolor    = red      % Colour of citations
}
\newcommand{\secQuickRef}[1]{\hyperref[#1]{\S\ref{#1}}}

\usepackage{array}
\usepackage{bbm}
\usepackage{subcaption} 
\usepackage{tabularx}
\usepackage{enumitem}
\usepackage{tikz}
\usetikzlibrary{positioning,arrows.meta,bending}

\usepackage{pagecolor} % DEBUG
\usetikzlibrary{decorations, decorations.text,backgrounds,decorations.pathreplacing,calligraphy} % DEBUG
\usetikzlibrary{arrows}
\usetikzlibrary{graphs}
\pgfdeclarelayer{bg}    % declare background layer
\pgfsetlayers{bg,main}  % set the order of the layers (main is the standard  layer)
\usepackage{booktabs}
\usepackage{xspace}

\renewcommand{\epsilon}{\varepsilon}

\usepackage[linesnumbered,ruled,noline]{algorithm2e}
\SetKwBlock{KwFunction}{function}{}
\SetKwBlock{KwOn}{on}{}
\SetKwBlock{KwPeriodically}{periodically}{}
\SetKwBlock{KwState}{state}{}
\SetKwBlock{KwAtomic}{atomically}{}
\SetKwBlock{KwConcurrently}{concurrently}{}
\SetCommentSty{textup} 
\SetKwComment{KwIttayComment}{(}{)}
\SetKw{KwAnd}{and}
\SetKw{KwOr}{or}
\SetKw{KwSetTimer}{setTimer}
\SetKwBlock{KwExpTimer}{on timer expiration}{}
\SetKwFor{ForAll}{forall}{}

\makeatletter
\let\oldnl\nl% Store \nl in \oldnl
\newcommand{\nonl}{\renewcommand{\nl}{\let\nl\oldnl}}% Remove line number for one line
\makeatother

\newcommand{\commentColor}{\color{gray}}

\usepackage{marginnote}

\ifdraft
\usepackage[colorinlistoftodos]{todonotes}
\newcounter{todocounter}

\usepackage[normalem]{ulem}

\newcommand{\ie}[2][Ittay]{\textbf{\color{green!70!black}[#2 -ie]}\marginnote{\color{green!70!black}\textup{*}}}

\fi

\usepackage{xparse}% For variable number of parameters http://ctan.org/pkg/xparse

\newcommand{\alg}{\ensuremath{ \mathcal{A} }\xspace}
\newcommand{\algoComment}[1]{{\hfill \color{gray}(#1)}}

\newcommand{\vTrue}{\ensuremath{ \textsf{True} }\xspace}
\newcommand{\vFalse}{\ensuremath{ \textsf{False} }\xspace}

\newcommand{\parent}[1]{\ensuremath{ \textit{parent}(#1) }\xspace}
\newcommand{\LOneRef}[1]{\ensuremath{ \textit{ref}_{A \to P}(#1) }\xspace}
\newcommand{\LOneResetRef}[1]{\ensuremath{ \textit{ref}_{A \to P}^{\:\textit{reset}}(#1) }\xspace}
\newcommand{\LTwoRef}[1]{\ensuremath{ \textit{ref}_{P \to A}(#1) }\xspace}
\newcommand{\blockParent}[1]{\ensuremath{ \textit{parent}(#1) }\xspace}
\newcommand{\blockID}[1]{\ensuremath{ \textit{id}(#1) }\xspace}
\newcommand{\isValid}{\ensuremath{ \textit{isValid} }\xspace}
\newcommand{\blockHeight}[1]{\ensuremath{ \textit{height}(#1) }\xspace}
\newcommand{\blockTime}[1]{\ensuremath{ t_\textit{gen}(#1) }\xspace}
\newcommand{\blockMembership}[1]{\ensuremath{ \textit{stakers}(#1) }\xspace}
\newcommand{\blockSet}{\ensuremath{ \mathcal{B} }\xspace}
\newcommand{\ELast}{\ensuremath{ B_\textit{last} }\xspace}
\newcommand{\EReset}{\ensuremath{ B_\textit{reset} }\xspace}

\newcommand{\deltaActive}{\ensuremath{ \Delta_\textit{active} }\xspace}

\newcommand{\deltaPrimaryWrite}{\ensuremath{ \Delta_{\primaryLedger\textit{-write}} }\xspace}
\newcommand{\deltaConsensus}{\ensuremath{ \Delta_\textit{consensus} }\xspace}
\newcommand{\gst}{\ensuremath{ t_\textit{GST} }\xspace}
\newcommand{\deltaPropagation}{\ensuremath{ \Delta_\textit{prop} }\xspace}
\newcommand{\progressStabilizationTime}{\ensuremath{ \Delta_\textit{stabilize} }\xspace}

\newcommand{\BCheckpoint}{\ensuremath{ B_\textit{checkpoint} }\xspace}
\newcommand{\bCheckpointed}{\ensuremath{ b_\textit{checkpointed} }\xspace}

\newcommand{\id}{\ensuremath{ \textit{id} }\xspace}
\newcommand{\committee}{\ensuremath{ \mathcal{N} }\xspace}

\newcommand{\inputFuncITK}[3]{v^{#1}_{#2}(#3)}

\newcommand{\consensusStepCore}{\ensuremath{ \textit{consensusStep} }\xspace}
\NewDocumentCommand{\consensusStep}{ooo}{%
    \IfNoValueTF{#1}{%
        \consensusStepCore
    }{%
        \IfNoValueTF{#2}{
            \ie{ERROR}
        }{%
            \IfNoValueTF{#3}{
                \ie{ERROR}
            }{
                \ensuremath{ \consensusStepCore(#1, #2, #3) }\xspace
            }%
        }%
    }%
}

\newcommand{\consensusValidateCore}{\ensuremath{ \textit{consensusValidate} }\xspace}
\NewDocumentCommand{\consensusValidate}{ooo}{%
    \IfNoValueTF{#1}{%
        \consensusValidateCore
    }{
        \IfNoValueTF{#2}{%
            \ie{ERROR}
        }{
            \IfNoValueTF{#3}{%
                \ie{ERROR}
            }{%
                \ensuremath{ \consensusValidateCore(#1, #2, #3) }\xspace
            }%
        }%
    }
}

\newcommand{\consensusProveDeviationCore}{\ensuremath{ \textit{consensusProveDeviation} }\xspace}
\NewDocumentCommand{\consensusProveDeviation}{oo}{\IfNoValueTF{#1}{\IfNoValueTF{#2}{\consensusProveDeviationCore}{ERROR}}{\ensuremath{ \consensusProveDeviationCore(#1, #2) }\xspace}}

\newcommand{\consensusThreshold}{\ensuremath{ \alpha }\xspace}

\newcommand{\ledger}{\ensuremath{ \mathcal{L} }\xspace}
\newcommand{\primaryLedger}{\ensuremath{ \underline{\mathcal{L}} }\xspace}
\newcommand{\bGenesis}{\ensuremath{ b_\textit{genesis} }\xspace}
\newcommand{\hash}{\ensuremath{ H }\xspace}

\let\oldcite\cite
\makeatletter
\def\cite#1{\oldcite{\zap@space#1 \@empty}}
\makeatother

\begin{document}

%\title{\ifdraft {\color{red} \textbf{DRAFT}}\\ \fi Aegis: A Decentralized Expansion Blockchain}
\title{\ifdraft {\color{red} \textbf{DRAFT}}\\ \fi Aegis: Tethering a Blockchain with Primary-Chain Stake}

% \author{Anonymous Submission}
% \affiliation{%
%   \city{}
%   \country{}
%   \institution{}
% }

\newcommand{\affilLightblocks}{{$^{\includegraphics[width=0.45em]{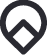}}$}}
\newcommand{\affilTAU}{$^{\ddagger}$}
\newcommand{\affilTechnion}{$^{\star}$}

\author{
    % Anonymized for review
Yogev Bar{-}On\affilTAU\affilLightblocks,
Roi Bar{-}Zur\affilTechnion\affilLightblocks,
Omer Ben{-}Porat\affilTechnion\affilLightblocks, 
Nimrod Cohen\affilLightblocks, \\
Ittay Eyal\affilTechnion\affilLightblocks, 
Matan Sitbon\affilLightblocks\\
\vspace{1.em}
\affilTAU Tel Aviv University \\
\affilLightblocks Lightblocks \\
\affilTechnion Technion\\
}
% \affilTAU Tel Aviv University \hfil
% \affilLightblocks Lightblocks \hfil
% \affilTechnion Technion

% \author{Yogev Bar{-}On}
% \affiliation{
%   \institution{Tel Aviv University}
%   \institution{Lightblocks Labs}
%   \country{}
% }
% \author{Roi Bar{-}Zur}
% \affiliation{
%   \institution{Technion}
%   \institution{Lightblocks Labs}
%   \country{}
% }
% \author{Omer Ben{-}Porat}
% \affiliation{
%   \institution{Technion}
%   \country{}
% }
% \author{Nimrod Cohen}
% \affiliation{
%   \institution{Lightblocks Labs}
%   \country{}
% }
% \author{Ittay Eyal}
% \affiliation{
%   \institution{Technion}
%   \country{}
% }
% \author{Matan Sitbon}
% \affiliation{
%   \institution{Lightblocks Labs}
%   \country{}
% }

% \renewcommand{\shortauthors}{Anon.}

\renewcommand{\shortauthors}{Bar{-}On et al.}

\begin{abstract} 
Blockchains implement decentralized monetary systems and applications.
Recent advancements enable what we call \emph{tethering} a blockchain to a \emph{primary} blockchain, securing the tethered chain by nodes that post primary-chain tokens as collateral.
The collateral ensures nodes behave as intended, until they withdraw it.
Unlike a \emph{Proof of Stake} blockchain which uses its own token as collateral, using primary-chain tokens shields the tethered chain from the volatility of its own token.

State-of-the-art tethered blockchains either rely on centralization, or make extreme assumptions: that all communication is synchronous, that operators remain correct even post-withdrawal, or that withdrawals can be indefinitely delayed by tethered-chain failures.

We prove that with partial synchrony, there is no solution to the problem. 
However, under the standard assumptions that communication with the primary chain is synchronous and communication among the tethered chain nodes is partially synchronous, there is a solution. 
We present a tethered-chain protocol called \emph{Aegis}.
Aegis uses references from its blocks to primary blocks to define committees, checkpoints on the primary chain to perpetuate decisions, and resets to establish new committees when previous ones become obsolete.
It ensures safety at all times and rapid progress when latency among Aegis nodes is low.
\end{abstract}


    


\maketitle

%%%%%%%%%%%%%%%%%%%%%%%%%%%%%%%%%%%%%%%%%%%%%%%%%%%%%%%%%%%
%%%%%%%%%%%%%%%%%%%%%%%%%%%%%%%%%%%%%%%%%%%%%%%%%%%%%%%%%%%
%%%%%%%%%%%%%%%%%%%%%%%%%%%%%%%%%%%%%%%%%%%%%%%%%%%%%%%%%%%

    \section{Introduction}

% Blockchains
Blockchains, such as Ethereum~\cite{wood2015yellow}, implement monetary systems and allow users to deploy \emph{smart contracts}~\cite{szabo1996smart} with arbitrary logic.
Users issue \emph{transactions} to exchange assets and interact with smart contracts.
Blockchain~\emph{nodes} aggregate transactions into blocks, forming a chain that determines the system's state. 

% New blockchains
As blockchain adoption grows, scalability has emerged as a critical challenge~\cite{chauhan2018blockchain,zhou2020solutions}.
Blockchains have limited throughput, often high latency, and are restricted to deterministic execution.
To address these limitations and implement new functionality, new blockchains are created~\cite{coinmarketcap2024}.
These new chains are either independent or~\emph{Layer-twos (L2s)} that expand existing blockchains%
\footnote{Not all L2s are blockchains by themselves, but many are~\cite{l2beat2024locked}.}.

% PoS
A popular approach to secure blockchains is \emph{Proof of Stake (PoS)}. 
In a PoS blockchain, nodes can \emph{stake}, namely, lock their tokens as collateral~\cite{kiayias2017ouroboros,daian2019snow,damato2022attacks}. 
Until they \emph{unstake}, namely, withdraw their collateral, they are \emph{active} nodes, participating in the creation of new blocks. 
If they misbehave, they are penalized via \emph{slashing}~\cite{buterin2015slasher}, losing all or part of their collateral.

% PoS is not secure for small chains. 
PoS works well for major blockchains that have reached a critical mass of participants and assets~\cite{coinmarketcap2024}.
Yet in smaller nascent systems, the token value is not established and depends on the blockchain's success, which in turn requires the blockchain to be secure.
However, a PoS blockchain can only be considered secure if the value of its collateral, expressed in its tokens, is high enough.
This chicken-and-egg problem poses an obstacle to newer chains, and remains a concern afterward, since token prices are often volatile~\cite{coinmarketcap2024}.

% Stake management
New systems like EigenLayer~\cite{eigenlayer} and Symbiotic~\cite{symbiotic} on Ethereum and Babylon~\cite{babylon2023bitcoinStaking} on Bitcoin~\cite{bitcoin2013protocol} enable staking tokens for securing external tasks.
Some of these, such as EigenLayer, implement~\emph{restaking}~\cite{durvasula2024robust}, allowing nodes to simultaneously secure both their native blockchain and additional tasks, while others, like Babylon, dedicate stake exclusively to new tasks.
These stake management systems represent a promising direction and have already garnered assets worth billions of US dollars~\cite{defillama2025restaking}.

% Tethering
We refer to~\emph{tethering} as using a~\emph{primary} blockchain for defining and coordinating a secondary chain's \emph{committee}, a set of nodes that create new blocks in the tethered chain using a \emph{consensus protocol}~\cite{castro1999practical}.
Stake management systems enable a particularly useful form of tethering, where the primary chain's tokens serve as collateral, shielding the tethered chain from the volatility of its own token.
Beyond addressing token volatility, tethering provides a practical solution for secure committee changes (epoch transitions), a notorious challenge in blockchain design~\cite{bano2019sok}.
Throughout this paper, we use stake-based terminology to discuss tethering, where nodes stake tokens to join committees and unstake to exit.
Like PoS blockchains, nodes are incentivized to be \emph{correct}, i.e., follow the protocol, under the threat of slashing, and may unstake at any time.

% Long term attacks
% Other than the volatility issue, smaller PoS systems do not have sufficient reputation to rely on long-term participation of established nodes, which is necessary to prevent long-range attacks~\cite{buterin2014weak,david2018praos}, where an adversary presents an alternative history, indistinguishable from the true one.
% A possible way to prevent long-range attacks is by checkpointing the tethered chain on the primary chain, as in~\cite{azouvi2022pikachu,tas2023bitcoin}.

\begin{figure*}[!ht]
    \centering
    \resizebox{0.98\textwidth}{!}{ % Fit figure to page width 
    \begin{tikzpicture}[
    every node/.style={rectangle, draw,minimum width=0.3cm, minimum height=.3cm},
    >=Stealth,
    bend angle=45, % Control the bend of curved arrows
    scale=1.0, % Control the bend of curved arrows
    transform shape=true,
]

\node[above, draw=none, anchor=west] at (-1,2.4) {Aegis chain};
\node[above, draw=none, anchor=west] at (-1,-0.5) {Primary chain};

% Aegis chain nodes
\node[draw=none, fill=none] (A0) at (-1,2) {...};
\node[fill=magenta!30] (A1) at (2,2) {$A1$};
\node[fill=gray!30] (A2) at (5.5,2) {$A2$};
\node[fill=green!30] (A3) at (6.8,2) {$A3$};
\node[fill=green!30] (A4) at (8,2) {$A4$};
\node[fill=green!30] (A5) at (9.5,2) {$A5$};
\node[fill=yellow!30] (A6) at (10.5,2) {$A6$};
\node[fill=cyan!30] (A7) at (15.5,2) {$A7$};
\node[draw=none, fill=none] (Ar) at (16.5,2) {...};
\node[draw=none, fill=none] (Adel) at (12.5,2) {$//$};

% Primary chain nodes
\node[draw=none, fill=none] (P0) at (-1,0) {...};
\node[fill=magenta!30] (P1) at (0,0) {$P1$};
\node[fill=gray!30] (P2) at (1,0) {$P2$};
\node[fill=white!30] (P3) at (2.5,0) {$P3$};
\node[fill=green!30] (P4) at (4,0) {$P4$};
\node[fill=yellow!30] (P5) at (8.5,0) {$P5$};
\node[fill=white!30] (P6) at (11.5,0) {$P6$};
\node[fill=cyan!30] (P7) at (14,0) {$P7$};
\node[fill=blue!30] (P8) at (15,0) {$P8$};
\node[draw=none, fill=none] (Pr) at (16.5,0) {...};
\node[draw=none, fill=none] (Pdel) at (12.5,0) {$//$};

\node[draw=none, fill=none] (delta) at (10.5,-1) {\Large $\deltaActive$};
\node[draw=none, fill=none] (Pdelta) at (8.5,-1) {};
\node[draw=none, fill=none] (Rdelta) at (13.5,-1) {};
\node[draw=none, fill=none] (Emptydelta) at (12.5,-1) {$//$};

% Draw arrows (modify as needed)
\draw[->, line width=0.8pt, dashed] (A1) -- (P1) node[midway, left, draw=none, scale=0.8, font=\Large] {reset};
\draw[->, line width=0.8pt,] (A1) -- (P2) node[midway, right, draw=none, scale=0.8, font=\Large] {ref};
\draw[->, line width=0.8pt,] (A2) -- (P4) node[midway, left, draw=none, scale=0.8, font=\Large] {ref};
\draw[->, line width=0.8pt,] (A3) -- (P4) node[midway, left, draw=none, scale=0.8, font=\Large] {ref};
\draw[->, line width=0.8pt,] (A4) -- (P4) node[midway, right, draw=none, scale=0.8, font=\Large] {ref};

\draw[->, line width=0.8pt,] (A5) -- (P5) node[midway, left, draw=none, scale=0.8, font=\Large] {ref};
\draw[->, line width=0.8pt,] (A6) -- (P5) node[midway, right, draw=none, scale=0.8, font=\Large] {ref};
\draw[->, line width=0.8pt, dashed] (A7) -- (P7) node[midway, left, draw=none, scale=0.8, font=\Large] {reset};
\draw[->, line width=0.8pt,] (A7) -- (P8) node[midway, right, draw=none, scale=0.8, font=\Large] {ref};

\draw[->, line width=0.8pt, bend right=20] (P3) to node[midway, right, draw=none, scale=0.8,xshift=0.0cm,yshift=-0.1cm,font=\Large] {chk} (A1);
\draw[->, line width=0.8pt, bend right=20] (P5) to node[midway, left, draw=none, scale=0.8,xshift=0.3cm,yshift=-0.45cm,font=\Large] {chk} (A3);
\draw[->, line width=0.8pt, bend right=20] (P6) to node[midway, right, draw=none, scale=0.8,xshift=0.1cm,yshift=-0.1cm, font=\Large] {chk} (A6);

% Connect upper chain nodes with arrows
\draw[<-] (A1) edge (A2) 
          (A2) edge (A3) 
          (A3) edge (A4) 
          (A4) edge (A5)
          (A5) edge (A6)
          (A6) edge (Adel)
          (A7) edge (Ar)
          (A0) edge (A1);

% Connect lower chain nodes with arrows
\draw[<-] (P1) edge (P2) 
          (P2) edge (P3)
          (P3) edge (P4) 
          (P4) edge (P5) 
          (P5) edge (P6)
          (P6) edge (Pdel)
          (P7) edge (P8)
          (P8) edge (Pr)
          (P0) edge (P1);

\draw[-|]{}
(Emptydelta) edge  (Rdelta)
(delta) edge  (Pdelta);

\draw[-]{}
(Pdel) edge (P7)
(Adel) edge (A7)
(delta) edge  (Emptydelta);

\end{tikzpicture}
    }
    \Description{Illustration of Aegis and primary-chain blocks.}
    \caption{Aegis and primary-chain blocks.}
    \label{fig:aegisTopology}
\end{figure*}

% Gap from prior work
Prior work~(\secQuickRef{sec:related}) extensively explored \emph{reconfiguration} of nodes in Byzantine Fault Tolerant (BFT) systems~\cite{lamport2010reconfiguring}.
A common approach is for an old committee to select the new committee before its tenure ends.
However, in our setting, the old committee might become deprecated before it can select the new committee:
When a node unstakes on the primary chain, this should trigger a reconfiguration, but if the committee fails to reach consensus in time, the node becomes unstaked and thus no longer incentivized to be correct.

State-of-the-art tethered chains, including recent work by Dong et al.~\cite{dong2024remote} and operational systems like Polygon~\cite{polygon2022polygon2}, Avalanche Subnets~\cite{avalanche2022subnets} and Cosmos~\cite{cosmos2025interchain}, have adopted the classical approach.
These systems either disallow unstaking until a successful reconfiguration completes—risking validators' collateral being locked indefinitely if reconfiguration fails—or rely on strong assumptions to avoid the failure scenario described above:
They either assume \emph{synchronous} message delivery with bounded time, ensuring committees reach consensus in time,
or assume nodes remain correct indefinitely even after unstaking, making further reconfiguration unnecessary.

Various Layer-2 solutions also aim to extend blockchain functionality and performance~\cite{gudgeon2020sok}.
Rollups batch transactions and post summaries to the primary chain~\cite{kalodner2018arbitrum,arbitrum2023anytrust,optimism2023block,starkware2023starkex,matter2023decentralized}, but rely on centralized sequencers and are bound by primary-chain finality times.
While recent work attempts to decentralize sequencers~\cite{motepalli2023sok,capretto2024fast}, it doesn't support dynamic committees.
Other approaches include sidechains~\cite{back2014enabling} that operate as separate blockchains connected to another blockchain through bridges for asset transfer, sharding~\cite{kokoris2018omniledger,wang2019sok} that modifies the primary chain's consensus to partition validators, and payment channels~\cite{poon2013lightning,decker2015duplex} that enable rapid off-chain bilateral transactions.
Unlike these approaches, tethered chains aim to be independent decentralized chains with their own rapid progress rate that derive security from primary-chain stake, without requiring modifications to the primary chain.

% Model 
To address the challenge of building a secure tethered chain, we first outline a model~(\secQuickRef{sec:model}).
Nodes can stake and unstake using the primary chain, but unstaking only completes after a delay~$\deltaActive$. 
Nodes with locked stake at the same time form a committee, which remains \emph{active} until any member completes unstaking.
As is standard~\cite{kiayias2017ouroboros,pu2023gorilla}, we assume that at all times a sufficient majority of nodes in each active committee are correct. 
The model employs a hybrid communication approach:
Nodes communicate asynchronously with each other until an unknown global stabilization time~$\gst$, after which message propagation is bounded, as in the~\emph{partial synchrony} model~\cite{dwork1988consensus}. 
However, at all times, nodes have synchronous (though slow) access to read and write to the primary chain~(\textit{cf.}~\cite{kelkar2020order}).

The choice of a hybrid communication model is not coincidental, but necessary to build a secure tethered chain~(\secQuickRef{sec:impossibility}).
To account for real-world conditions, where there may be potential outages in internode communication, especially in nascent systems, we assume partial synchrony.
However, we prove that if communication with the primary chain is also partially synchronous, no solution can be secure.
To circumvent this impossibility, we bound the access-time to the primary chain, similarly to many existing L2s~\cite{optimism2023block,arbitrum2023anytrust}.
Nevertheless, this bound can be chosen to be sufficiently large (possibly in the order of hours or days), and should not affect the progress rate of the tethered chain.

% Aegis 
We present \emph{Aegis}%
\footnote{Named after the powerful divine shield carried by Zeus and Athena in Greek mythology, which first appears in Homer's Iliad.}%
~(\secQuickRef{sec:aegis}), a tethered blockchain.
Each \emph{Aegis block} references both its predecessor and a corresponding \emph{primary block}.
The committee defined by that primary block is the committee in charge of generating the next Aegis block using a consensus protocol.
While maintaining the Aegis chain, its nodes occasionally checkpoint the latest block with a dedicated smart contract on the primary-chain.

In the happy flow, the checkpoints are less than~$\deltaActive$ apart, so there is always an active committee.
However, creating an Aegis block might take too long due to internode communication delays.
If it takes longer than~$\deltaActive$, there would be nothing to checkpoint. 
In this case, any Aegis node can issue a \emph{reset} on the primary chain to specify a new committee. 
At this point, the committee is defined in the primary-chain reset, but not yet in any Aegis block, until one is generated.

Figure~\ref{fig:aegisTopology} illustrates the topology of the Aegis chain ($A1, A2, \dots$) and the primary chain ($P1, P2, \dots$). 
Most Aegis blocks ($A2$--$A6$) are generated by the committee referenced by their ancestor signified by the same color. 
Some are checkpointed in primary blocks ($P3, P5, P6$). 
If a checkpoint cannot be placed in time (and before the first Aegis block), a reset is issued ($P1, P7$), and the subsequent Aegis block is generated by the committee of the reset block ($A1, A7$). 

There is a risk that the committee defined by the reset and the one defined by the previous block both generate a block at the same height (distance from chain root), forming a so-called \emph{fork}. 
To prevent that, both the protocol and the primary-chain contract enforce timing constraints, avoiding consensus decisions that cannot be checkpointed and premature resets. 
 
The Aegis protocol thus comprises the primary-chain smart contract, which handles checkpoints and resets, and a distributed protocol executed by Aegis nodes, which invokes a classical single-committee consensus protocol to extend the Aegis chain.

% Security 
To prove the security of the protocol~(\secQuickRef{sec:security}) we show that correct Aegis nodes never disagree on the blockchain content at any height. 
First, we use backward induction to show that if a block is created by an active committee, then all of its ancestors were created by active committees. 
This holds because this is verified for every block either by its child's committee or by the primary blockchain in a checkpoint. 
Then, using forward induction, we show that if all nodes agree up to some height, a classical consensus protocol guarantees the necessary properties for the next block. 
Similar arguments show validity, namely that if all nodes are correct and have the same input value, no other value is logged. 
We then show that after~$\gst$ the protocol guarantees an active committee forms and extends the blockchain. 
Straightforward indistinguishability arguments bridge the gap between the static committee of the standard consensus protocol and the dynamic committees of Aegis. 

% Practical considerations and extensions
We then discuss practical considerations for performance and security, as well as potential extensions of our tethering paradigm~(\secQuickRef{sec:practicalConcerns}). 
While our implementation focuses on stake-based committee management, the reset mechanism can be adapted to support different committee transition logics, such as majority decisions or after a predetermined number of blocks.

Aegis can be readily implemented to deploy various decentralized services as tethered chains.
These include scaling solutions, like rollups;
and solutions providing new functionality, like oracles and randomness sources.
% For practical applications, Aegis is almost directly applicable, subject to some practical considerations (\secQuickRef{sec:practical}). 
% For example, efficiency considerations affect checkpointing frequency; and restaking of Ethereum stake is possible with Eigen Layer, allowing to quickly onboard large staking amounts. 

%%%%%%%%%%%%%%%%%%%%%%%%%%%%%%%%%%%%%%%%%%%%%%%%%%%%%%%%%%%
%%%%%%%%%%%%%%%%%%%%%%%%%%%%%%%%%%%%%%%%%%%%%%%%%%%%%%%%%%%
%%%%%%%%%%%%%%%%%%%%%%%%%%%%%%%%%%%%%%%%%%%%%%%%%%%%%%%%%%%

    \section{Related Work} \label{sec:related}

%%%%%%%%%%%%%%%%%%%%%%%%%%%%%%%%%%%%%%%%%%%%%%%%%%%%%%%%%%%
%%%%%%%%%%%%%%%%%%%%%%%%%%%%%%%%%%%%%%%%%%%%%%%%%%%%%%%%%%%

\paragraph{Classical Reconfiguration.}
The reconfiguration problem has received intensive treatment in the classical distributed-systems literature. 
Some solutions rely on the old committee completing a full state handoff before leaving \cite{lamport2009vertical}, or finalizing and aborting any pending entries \cite{aublin2015next}; both assumptions fail if nodes exit prematurely.
Other approaches store consensus decisions in a separate reliable replica \cite{abraham2016bvp,castro1999practical,kotla2007zyzzyva,whittaker2021matchmaker}, incurring frequent reads and writes to the replica \cite{castro1999practical,kotla2007zyzzyva,whittaker2021matchmaker}, which implies additional latency.
In our setting, a cooldown period for unstaking provides a bounded window to contact the outgoing committee via the primary chain, eliminating the need for continuous access to a reconfiguration service.
Another proposal requires deprecated committees to destroy old private keys \cite{kuznetsov2022asynchronous}, but this can't be enforced in a stake-based model.

Previous work on state machine replication we are aware of sufficed with external validity. 
That is, orders that are correctly formatted are considered valid, and the only requirement is that such orders are eventually placed in the log. 
This allows values to be aggregated as well as submitted long before they are placed in the blockchain~\cite{castro1999practical}. 
This assumption is insufficient for blockchain applications, where various, possibly competing users submit transactions whose content depends on the latest system state. 
We therefore introduce a stronger validity requirement called no-batching validity. 

%%%%%%%%%%%%%%%%%%%%%%%%%%%%%%%%%%%%%%%%%%%%%%%%%%%%%%%%%%%
%%%%%%%%%%%%%%%%%%%%%%%%%%%%%%%%%%%%%%%%%%%%%%%%%%%%%%%%%%%

\paragraph{Tethered Chains.}
Tethered chains are blockchains that use tokens from a primary chain as collateral, simplifying bootstrapping compared to starting a new blockchain and shielding them from the volatility of their own token.
At first glance, tethered chains may resemble classical State Machine Replication with dynamic reconfiguration \cite{lynch1996distributed,lamport2010reconfiguring}, where the committee for block~$k$ selects the committee for block~$k+1$.
However, if the creation of block~$k+1$ is delayed, members of its committee may unstake in the meanwhile and lose their incentive to behave correctly.
Aegis addresses this by having each block reference its next committee directly from the primary chain and allowing a reset if the current committee stalls.

Existing tethered systems, including Avalanche subnets \cite{avalanche2022subnets}, Polygon~2.0 \cite{polygon2022polygon2}, and Cosmos \cite{cosmos2025interchain}, face this fundamental issue.
They either assume synchronous communication to avoid delayed block creation, or assume nodes remain correct indefinitely so delayed blocks are still handled correctly.
New tethered chains can be created using stake management solutions such as EigenLayer \cite{eigenlayer} and Symbiotic \cite{symbiotic} on Ethereum or Babylon \cite{babylon2023bitcoinStaking} on Bitcoin, these platforms leave open how to implement the chain.

BMS \cite{steinhoff2021bms} proposes a trusted reconfiguration platform that must approve each committee change, which becomes impractical if a node wants to exit but the tethered chain fails to grant timely approval.
By contrast, Aegis enables nodes to withdraw unilaterally without compromising security, thanks to primary-chain collateral tracking.

An alternative approach, merge-mining \cite{bonneau2015sok}, applies to Proof-of-Work systems \cite{nakamoto2008bitcoin,sompolinsky2021phantom} by reusing computational power across blockchains, but provides no penalty on the main chain for misbehavior on the tethered chain.
Because Aegis slashes misbehaving nodes at the primary-chain level, it preserves correct incentives even after committees become deprecated.

%%%%%%%%%%%%%%%%%%%%%%%%%%%%%%%%%%%%%%%%%%%%%%%%%%%%%%%%%%%
%%%%%%%%%%%%%%%%%%%%%%%%%%%%%%%%%%%%%%%%%%%%%%%%%%%%%%%%%%%

\paragraph{Long-Range Attacks and Checkpoints.}
Long-range attacks occur when an adversary presents an alternative history that appears indistinguishable from the genuine chain, particularly to newcomers lacking the full chain state~\cite{buterin2014long}.
These attacks enable double-spending or can disrupt consensus.
They can affect both PoW and PoS blockchains, though in PoW an attacker must hold a majority of the network's computational power (a 51\% attack \cite{bano2019sok}).
In PoS, however, only a majority of stake at some historical point is required, and once nodes withdraw, they are no longer at risk of slashing.
As a result, long-range attacks are often considered more severe in PoS settings.

Prior work has explored \emph{checkpointing} to mitigate such attacks.
Early methods relied on \emph{social consensus}, wherein an external committee vouches for the current chain state \cite{barber2012bitter,buterin2014weak,daian2019snow}.
Later studies employed BFT protocols for checkpoint creation \cite{rana2022optimal,sankagiri2021blockchain}, or used separate PoW networks to timestamp checkpoints \cite{karakostas2021securing,azouvi2022pikachu,tas2023bitcoin}.
Because Aegis already depends on the primary chain for stake management, it reuses that same chain to record checkpoints, avoiding extra trust assumptions.
Another approach, Winkle \cite{azouvi2020winkle}, suggests weighting checkpoints by the volume of subsequent transactions.

%%%%%%%%%%%%%%%%%%%%%%%%%%%%%%%%%%%%%%%%%%%%%%%%%%%%%%%%%%%
%%%%%%%%%%%%%%%%%%%%%%%%%%%%%%%%%%%%%%%%%%%%%%%%%%%%%%%%%%%

\paragraph{Layer-2 Solutions.}
Various Layer-2 (L2) solutions aim to extend the performance and functionality of a primary chain (also called Layer-1 or L1) but differ fundamentally in their approach and limitations.
\emph{Rollups} improve transaction throughput by batching transactions in a separate ledger and occasionally posting summaries to the L1 for committing changes.
Rollups come in two forms:
\emph{Optimistic} rollups~\cite{kalodner2018arbitrum,arbitrum2023anytrust,optimism2023block} rely on a dispute process in the L1 to correct, and possibly penalize, incorrect summaries; while~\emph{zero-knowledge} rollups~\cite{starkware2023starkex,matter2023decentralized} generate cryptographic proofs of summary validity.
Gudgeon et al.~\cite{gudgeon2020sok} provide a comprehensive systemization of knowledge.

A key limitation of rollups is their reliance on centralized components for transaction ordering (via \emph{sequencers}), which poses a risk of transaction censorship and manipulation for financial gain.
Recent work has explored decentralizing the sequencer~\cite{motepalli2023sok}, with Setchain~\cite{capretto2024fast} implementing Byzantine-tolerant transaction ordering, but its use of static committees limits its ability to handle dynamic stake changes.
Another limitation is that transaction finality is bound to block confirmation times in the primary chain, as changes are only finalized once a summary is finalized in the primary chain.
Aegis overcomes both limitations with its decentralized consensus and rapid finality:
Blocks can be accepted once they are created, which may be significantly sooner than until they are checkpointed.
Furthermore, Aegis can serve as a decentralized sequencer for rollups with dynamic committee reconfiguration.

\emph{Sidechains}~\cite{back2014enabling,polygon2023pos,gnosis2023chain} operate as independent blockchains that connect to another blockchain through~\emph{bridges}, protocols that enable asset transfers between chains.
These are orthogonal to tethered chains, as a tethered chain can also be a sidechain if it allows asset transfer.

\emph{Sharding}~\cite{kokoris2018omniledger,wang2019sok} takes a more invasive approach by modifying the core consensus protocol of the L1 to partition validators across multiple chains.
Unlike Aegis, which operates independently while deriving security from the primary chain, sharding requires fundamental changes to the L1 design.

Payment channels are point-to-point channels~\cite{poon2013lightning, decker2015duplex} that can aggregate several payments between two parties and form payment networks.
While they enable efficient transactions between parties, they fundamentally differ from Aegis's approach as they rely on bilateral relationships rather than decentralized consensus, with each channel exclusively controlled by its two counterparties.

Our goal differs from these approaches:
We seek to create a chain with its own rapid progress rate, which relies on a primary chain for managing its stake;
other than this dependency, it may not have any other relation to the primary chain.

%%%%%%%%%%%%%%%%%%%%%%%%%%%%%%%%%%%%%%%%%%%%%%%%%%%%%%%%%%%
%%%%%%%%%%%%%%%%%%%%%%%%%%%%%%%%%%%%%%%%%%%%%%%%%%%%%%%%%%%

\paragraph{Chain Combining.}
% % Related idea: Ledger-bootstrapping
Several previous works combine multiple ledgers to create a ledger with better properties~\cite{fitzi2020ledger,wang2022trustboost,tas2024circuit}. 
However, all steps are taken within the used ledgers, unlike Aegis, whose tethered chain can progress by tethered-chain consensus steps faster than the rate it is checkpointed.

%%%%%%%%%%%%%%%%%%%%%%%%%%%%%%%%%%%%%%%%%%%%%%%%%%%%%%%%%%%
%%%%%%%%%%%%%%%%%%%%%%%%%%%%%%%%%%%%%%%%%%%%%%%%%%%%%%%%%%%
%%%%%%%%%%%%%%%%%%%%%%%%%%%%%%%%%%%%%%%%%%%%%%%%%%%%%%%%%%%

    \section{Model} \label{sec:model}

The system~(\S\ref{sec:model:elements}) includes a primary blockchain and a set of nodes. 
We assume a generic consensus algorithm~(\S\ref{sec:model:consensus}) that the nodes use aiming to form a tethered blockchain~(\S\ref{sec:model:goal}). 
Appendix~\ref{app:notation} summarizes the notation.

%%%%%%%%%%%%%%%%%%%%%%%%%%%%%%%%%%%%%%%%%%%%%%%%%%%%%%%%%%%
%%%%%%%%%%%%%%%%%%%%%%%%%%%%%%%%%%%%%%%%%%%%%%%%%%%%%%%%%%%

    \subsection{Principals and Network} \label{sec:model:elements}

% nodes and time: 
The system comprises an unbounded set of nodes (as in~\cite{kiayias2017ouroboros,pu2023gorilla}) and a \emph{primary blockchain} (or simply \emph{primary}) denoted by $\primaryLedger$.
Time progresses in steps. 
In each step a node receives messages from previous steps, executes local computations, updates its state, and sends messages. 

% L1 blockchain: 
The primary chain is implemented by a trusted party that maintains a state including token balance for each node and a so-called \emph{staked amount} for each node. 
Initially, each node has some initial token balance. 
The primary also maintains node-defined \emph{smart contracts} implementing arbitrary automatons. 
Nodes interact with the primary by issuing transactions that update its state and interact with smart contracts. 
The primary chain aggregates transactions into \emph{blocks} and extends the blockchain at set intervals. 
Its state is the result of parsing the series of all transactions. 
A node cannot guess the hash of a future block except with negligible probability. 

A node can issue a transaction to \emph{stake}, deposit some of their tokens, or \emph{unstake}, order the withdrawal of tokens she had previously staked. 
Unstaking completes~$\deltaActive$ steps after the order transaction is placed in the chain. 
While a node's tokens are staked, the node is considered \emph{active}. 
During this period, a node can either follow the protocol and be considered \emph{correct}, or misbehave arbitrarily.
Correct nodes adhere to the protocol while they are staked, if misbehavior could result in losing their stake.
However, once they unstake, they are no longer bound by the protocol and can behave arbitrarily without the risk of losing their stake.

The set of nodes that have staked but not yet unstaked after a specific block is called a \emph{committee}; until the time any of them completed unstaking, the committee is \emph{active}. 
In all committees, the ratio of stake held by correct nodes is greater than a threshold, denoted by~$\consensusThreshold$.

% Communication: 
Communication among nodes is reliable and partially synchronous: 
In every execution there exists a time~$\gst$, unknown to the nodes, and message delivery is bounded only after it. 
That is, if a message is sent at time~$t$ then it arrives at all nodes by time~${\max(t, \gst) + \deltaPropagation}$. 
In particular, blocks published by the nodes are delivered to all nodes, even those that stake much later%
\footnote{In practice~\cite{bitcoin2013protocol}, this is implemented by the nodes' peer-to-peer network, where nodes keep all published blocks available. 
%, thus this is not a DoS vulnerability.
}. 

Communication with the primary chain is synchronous: All nodes can observe the current state of the primary chain and issue transactions that are added to it within a bounded time~$\deltaPrimaryWrite$.
This assumption is stronger than the one for internode communication, but it is more reasonable in this case because we rely on an already established primary chain.
In practice, this assumption is met by choosing conservative bounds for the time it takes to write to the primary chain, even in times of congestion in the primary chain and perhaps even in the face of a DoS attack.

System behavior is orchestrated by a \emph{scheduler}. 
The scheduler determines message arrival times, subject to the constraints above. 
It also determines when nodes stake and unstake--to represent the fact this decision is exogenous to the protocol that the participating nodes execute. 

%%%%%%%%%%%%%%%%%%%%%%%%%%%%%%%%%%%%%%%%%%%%%%%%%%%%%%%%%%%
%%%%%%%%%%%%%%%%%%%%%%%%%%%%%%%%%%%%%%%%%%%%%%%%%%%%%%%%%%%

        \subsection{Consensus Algorithm} \label{sec:model:consensus}

We are given a (single-instance) Byzantine fault tolerant consensus protocol designed for partial synchrony (e.g.,~\cite{dwork1988consensus,castro1999practical}) implemented in a function $\consensusStep$. 
Called by a process, $\consensusStep$ takes steps in the protocol, maintaining a local state as necessary. 
It is executed by a weighted set of nodes $\committee$, given as a parameter. 
(We will later use the staked amount as the weight, thus $\committee$ will constitute a committee.) % such that the nodes and their weights are known to all. 
It takes a consensus instance ID, which is a pair of the set of nodes $\committee$ and a nonce \id to distinguish different instances. 
% Thus, each consensus instance is identified by the pair~$(\committee, \id)$. 
It also takes a block~$b$ the caller wishes to propose. 
It returns either a block if it reached a decision or~$\bot$ otherwise: 
$\consensusStep(b; \committee, \textit{id}) \rightarrow \textit{block} \textup{ or } \bot$.

If starting at~$t_0$, the nodes in $\committee$ call $\consensusStep$ in every step~$t \ge t_0$ with the same~$(\committee, \id)$ and the aggregate weight of correct nodes (those who follow the protocol) is greater than~$\consensusThreshold$, then the following guarantees hold.
% Denote by~$t_0$ the first step where a node calls $\consensusStep$ with some~$(\committee, \id)$. 
% If there exists a subset of~$\committee$ with aggregate weight greater than~$\consensusThreshold$ that is correct and calls $\consensusStep$ in every step~$t \ge t_0$, then the following guarantees hold. 
% in every time step~$t \ge t_0$ the subset of~$\committee$ with 
% a set of nodes execute $\consensusStep$ and the aggregate weight of correct nodes is greater than~$\consensusThreshold$, then the following guarantees hold. 
\begin{description}

\item[Agreement] For all correct nodes~${i, j}$ (maybe~$i = j$) and times~$t, t'$ (maybe~${t = t'}$), if node~$i$'s $\consensusStep$ call at step~$t > t_0$ with ID $(\committee, \textit{id})$ returns~$b \neq \bot$ and node~$j$'s $\consensusStep$ call at step~$t' > t_0$ with ID $(\committee, \textit{id})$  returns~$b' \neq \bot$, then~$b = b'$. 

\item[Validity] If all nodes in~$\committee$ are correct, and they all execute $\consensusStep$ with the same ID and the same proposal~$b$, then their $\consensusStep$ call never returns a non-$\bot$ value~$b' \neq b$. 

\item[Termination] 
% If~${\deltaConsensus + \deltaPrimaryWrite < \deltaActive}$ then after~$\gst$ all correct nodes that execute $\consensusStep$ return a block within~$\deltaConsensus$ steps. 
% \IE{Why is this limitation necessary, why not simply guarantee termination by~$\deltaConsensus$?}
For all correct nodes~$i$ and times~$t > \max\{t_0, \gst\} + \deltaConsensus$, node~$i$'s $\consensusStep$ call with ID $(\committee, \textit{id})$ returns a non-$\bot$ value. 

\end{description}

The agreement property implies that a block, which a correct node has voted for, cannot be overturned.
Protocols that do not satisfy this property are out of scope.

We require the given BFT consensus protocol to include two additional functions.
The function \consensusValidate takes a block and an ID (including the set~$\committee$) and returns \vTrue if and only if the block is the result of those nodes, with a correct ratio above~$\consensusThreshold$, running $\consensusStep$. 
(In practice this can be implemented with cryptographic signatures.) 

Finally, the protocol supports \emph{forensics}, that is, it allows identifying misbehavior of the nodes~\cite{sheng2021bft, civit2021polygraph}. 
Specifically, if the number of Byzantine nodes is larger than the bound, they can cause an agreement violation. 
In such a case, the correct nodes can publish data allowing them to identify misbehaving parties.
To allow for enough time to run forensics before nodes in an active committee can unstake, we require that~${\deltaActive > 3\deltaPrimaryWrite}$.

%%%%%%%%%%%%%%%%%%%%%%%%%%%%%%%%%%%%%%%%%%%%%%%%%%%%%%%%%%%
%%%%%%%%%%%%%%%%%%%%%%%%%%%%%%%%%%%%%%%%%%%%%%%%%%%%%%%%%%%

        \subsection{Goal} \label{sec:model:goal}

The goal of the system is to form a tethered blockchain, denoted by~$\ledger$.
Each node~$i$ has a time-varying input value for each position of~$\ledger$.
For a height~$k$ (distance from the root) and time~$t$, it is denoted by~$\inputFuncITK{i}{t}{k}$.
We do not consider the content of~$\ledger$, which could serve as a decentralized sequencer, oracle, etc.\ (see~\secQuickRef{sec:related}). 
We only consider its blocks, which are maintained by the nodes as local vectors. 
If a node stores a block in a position~$k$ in its local vector we say it \emph{logs} this block. 
Note that this is a chain, not a single-instance consensus as in Section~\ref{sec:model:consensus}. 
The tethered blockchain should achieve the following 3 properties (referring to tethered-chain blocks and heights).

We use the standard Agreement property. 

\begin{description}
\item[Agreement] If a correct node~$i$ logs a block~$b$ in position~$k$ at time~$t$ and a correct node~$j$ logs block $b'$ in position~$k$ at time~$t'$ (maybe~$i=j$ and maybe~$t=t'$) then~$b=b'$. 
\end{description}

Agreement implies instant finality:
Once a value is logged, it will not be overturned.

For validity, we require each decision to be based on fresh inputs. 
This embodies the interactive nature of the protocol, where each block is based on the previous one. 
It also prevents batching outside block limits, that is, the nodes cannot agree in a single consensus iteration on multiple blocks. 

\begin{description}
\item[No-Batching validity] Let~$t'$ be the time of the first decision at height~$k$ (i.e., at least one node decides at height~$k$ at time~$t'$, and no node decides at height~$k$ before~$t'$).
If
(1) all input functions for height~$k+1$ return value~$v$ after~$t'$ (i.e., $\forall i, t'' > t': \inputFuncITK{i}{t''}{k+1} = v$);
(2) all nodes are correct;
and
(3) there are no configuration changes;
then no node decides on a value different from~$v$ at height~$k + 1$.
% \item[No-Batching validity] If at time~$t'$ the first decision(s) on height~$k$ occurs (at least one node decides on height~$k$ at~$t'$ and in all earlier times no node decides on height~$k$), and subsequently the input functions of all nodes return~$v$ ($\forall i, t'' > t': \inputFuncITK{i}{t''}{k+1} = v$), no node is incorrect and there are no configuration changes, then no node decides on a value different from~$v$ at height~$k + 1$. 
% If all nodes are correct and their inputs for all positions~${v_1, v_2, \dots}$ are the same, then in all positions~$k$ no node logs a block with a different value $u_k \neq v_k$. 
\end{description}

Finally, we require rapid progress after the protocol stabilizes at some time after~$\gst$. 

\begin{description}
\item[Progress] 
There exists a protocol stabilization bound~$\progressStabilizationTime$ such that for all executions, for all~${t > \gst + \progressStabilizationTime}$, if the latest decided value by any node at~$t$ is~$k$, then by~$t + \deltaConsensus + \deltaPropagation$ all correct nodes decide on~${k + 1}$
\end{description}

% (1) For all~$k$ no two nodes decide on different values for height~$k$. 
% (2) If at time~$t'$ any node decides for the first time on height~$k$, and for all times~$t'' > t'$, the input functions of all nodes return~$v$, no node is incorrect and there are no configuration changes, then no node decides on a value different from~$v$ at height~$k + 1$. 
% (3) If the latest decided value by any node at~$t'$ is k, then by~$t' + \deltaConsensus$ all nodes decide on~${k + 1}$. 

When it is not clear from context, we refer to the blockchain properties as blockchain-agreement etc., and the consensus properties as consensus-agreement etc.

%%%%%%%%%%%%%%%%%%%%%%%%%%%%%%%%%%%%%%%%%%%%%%%%%
%%%%%%%%%%%%%%%%%%%%%%%%%%%%%%%%%%%%%%%%%%%%%%%%%
%%%%%%%%%%%%%%%%%%%%%%%%%%%%%%%%%%%%%%%%%%%%%%%%%

\section{Bounded Primary-Chain Access is Necessary}\label{sec:impossibility}

We now show that Aegis is optimal in the following sense: There is no solution to the problem without synchronous access to the primary ledger. 
Formally, assume for the rest of this section that instead of synchronous access to the primary ledger, there are no guarantees before~$\gst$. 
That is, a write order to the primary at time~$t$ is only guaranteed to be completed by time~${\max\{t, \gst\} + \deltaPrimaryWrite}$.
Note that all messages ever sent and all primary-writes ever completed are visible to newly joining nodes.

In addition, for the sake of simplicity, in the rest of the section, we assume that solutions can only log one bit, either ~0 or~1.
Our results, trivially transfer to the case where more values are possible.

% Note that in practice staking/unstaking can happen multiple times. We overcome this in Aegis. But even without this, no solution is possible, as we show in this section. 

        \subsection{Naive Attempts}

For grasping the problem at hand, we suggest a few naive attempts which fail to achieve the desired properties.

One possible approach for a solution is for each node to submit to the primary ledger a proposal for each height with incrementing order. 
The first value at each height is the new bit, so all nodes log the same values. 
Only this extends at most at a rate of $\deltaPrimaryWrite$ even after $\gst$, since this is the time (bound) to find the block at each height (the time to write a block to the ledger);
and doesn't satisfy the Progress property. 

Another possibility is for nodes to submit proposals for more than one specific height to every primary block, and then nodes would pick the first proposal for each height.
Thus, allowing nodes to log proposals at a higher but constant rate.
This may satisfy progress, but doesn't satisfy No-Batching Validity, as nodes have to propose several blocks at once, and proposals may contain transactions which depend on previous ones to be accepted.

Another option is to switch committees at set-length epochs, say, 1000 blocks (assuming $\deltaActive$ is large enough). 
Before the end of an epoch, the committee writes back the last value, and the next committee picks up from that point. 
Again, this does not fulfill Progress property, since during the handoff the chain must stop for about~$O(\deltaPrimaryWrite)$ until the previous-epoch value is written.

%%%%%%%%%%%%%%%%%%%%%%%%%%%%%%%%%%%%%%%%%%%%%%%%%
%%%%%%%%%%%%%%%%%%%%%%%%%%%%%%%%%%%%%%%%%%%%%%%%%

        \subsection{Impossibility Proof} 

First, before we show the impossibility, we show any solution which satisfies No-Batching Validity must allow no batching.

\begin{lemma}\label{lem:noBatching}
Let~$\alg$ be a protocol, and let~$\sigma$ be an execution of~$\alg$ where all nodes are correct.
Let~$t_k$ be the time when the first node decides on height~$k$, and let~$t_{k+1}$ be the time when the first node decides on height~$k+1$.
If~$\alg$ satisfies No-Batching Validity, then~$t_{k+1} > t_k$.
\end{lemma}

\begin{proof}
Let~$\alg$ be a protocol that satisfies No-Batching Validity, and let~$\sigma$ be an execution of the protocol where all nodes are correct.
Let~$t_k$ be the time when the first node decides on height~$k$, and let~$t_{k+1}$ be the time when the first node decides on height~$k+1$.
Assume the value decided for height~$k + 1$ is~$v$.

Assume towards contradiction that~$t_{k+1} \le t_k$.
Consider an execution~$\sigma'$ identical to~$\sigma$ up to~$t_{k+1}$, and afterward, all input functions for height~$k+1$ and times~${t'' > t_k}$ return~${v' \neq v}$, all nodes are correct, and there are no configuration changes.
By No-Batching Validity, no node decides on a value different from~$v'$ at height~$k + 1$.
But~$v \neq v'$ by definition, so this contradicts the assumption that~$t_{k+1} \le t_k$.
\end{proof}

We are now ready to show that no solution is possible without synchronous access to the primary ledger.

\begin{theorem}
No protocol implements a tethered chain with partially synchronous ledger writing, when facing Byzantine behavior of inactive (unstaked) nodes.
\end{theorem}

\begin{proof}
Assume towards contradiction that there is a protocol~$\alg$ satisfying Agreement, No-Batching Validity, and Progress.

The proof uses a set of executions and indistinguishability arguments. 
Figure~\ref{fig:impossibility} illustrates the executions. 

Let~$\sigma$ be an execution of~$\alg$ where all nodes are correct.
Let~$\gst$ be the time from which network communication becomes stable in the execution~$\sigma$.
As in the Progress property, the protocol starts progressing rapidly after time~$\gst + \progressStabilizationTime$.
Denote by~$k$ the largest height decided by any node by that time. 
Define $t^{\textit{divergence}}$ as the earliest time any node logs height~$k+1$.
Since no node can log height~$k + 1$ before completing height~$k$~(Lemma~\ref{lem:noBatching}), we have $t^{\textit{divergence}} > \gst + \progressStabilizationTime$.

We now define two executions, $\sigma^0$ and $\sigma^1$, which follow $\sigma$ exactly until $t^{\textit{divergence}}$ but diverge afterward.
In $\sigma^0$, all input functions for height~$k+2$ return~0 after $t^{\textit{divergence}}$.
In $\sigma^1$, they all return~1 after $t^{\textit{divergence}}$.
We assume that no committee changes occur after $t^{\textit{divergence}}$ in either execution, that all nodes remain correct and that both executions~$\gst$ is the same as in $\sigma$.
We also assume that the scheduler delivers all node-to-node messages within a single step and completes primary writes at the latest possible time, within~$\deltaPrimaryWrite$ steps of each write. 
By Progress, in both executions~$\sigma^0$ and~$\sigma^1$, at least one node decides on height~$k+2$ by time~$t^{\textit{decide}} = t^{\textit{divergence}} + \deltaConsensus + \deltaPropagation$.
By No-Batching Validity, in $\sigma^0$ one node logs~0, while in $\sigma^1$ one node (not necessarily the same one) logs~1.

Any write to the primary ledger between $t^{\textit{divergence}}$ and $t^{\textit{decide}}$ can take up to $\deltaPrimaryWrite$ steps, and we assume $\deltaPrimaryWrite$ is large enough compared to $\deltaConsensus + \deltaPropagation$
so that none of these writes occur by time $t^{\textit{decide}}$.

Next, define a third execution, $\sigma^{1'}$, which follows $\sigma^1$ until $t^{\textit{decide}}$ and then diverges.
All nodes in the old committee issue unstake transactions at time~${t^{\textit{decide}}}$ and a disjoint set of nodes stakes in to form a new committee at time~${t^{\textit{decide}} + \deltaActive}$.
The scheduler delivers the unstaking and staking transactions as soon as possible.
Set $\gst' = t^{\textit{decide}} + \deltaActive$ for the execution~$\sigma^{1'}$.
The scheduler delays all messages from the old committee to the new committee and all message writes until the latest possible time~$\gst' + \deltaPropagation$, and $\gst' + \deltaPrimaryWrite$ respectively.
By Progress, the new committee, now active, logs a value at height~$k+2$ by some time~$t^{\textit{decide}'} > t^{\textit{decide}} + \deltaActive$;
by Agreement, that value is~1.

Finally, we construct $\sigma^x$ by combining the prefix of $\sigma^0$ up to $t^{\textit{decide}}$ with the reconfiguration at time~$T$ that occurs in $\sigma^{1'}$.
In $\sigma^x$, the old committee has already logged~0 for height~$k+2$ by $t^0$.
We set the stabilization time of $\sigma^x$ to $\gst^x = t^{\textit{decide}'}$, and let the scheduler defer all old-committee messages and ledger writes which occur after~$t^{\textit{decide}}$ until $\gst^x$, except for the following.
Since by $T + \deltaActive$ ($< \gst^x$) the old committee is no longer active, it can behave arbitrarily without being slashed.
At that point, it produces the same messages and write orders as it did in $\sigma^{1'}$, which the scheduler delivers at the corresponding times they appeared in $\sigma^{1'}$.
Hence, the new committee in $\sigma^x$ sees exactly the events up until~$t^{\textit{decide}'}$ as in $\sigma^{1'}$, so it cannot distinguish $\sigma^x$ from $\sigma^{1'}$ and therefore logs~1 at height~$k+2$ at time~$t^{\textit{decide}'}$.
This contradicts Agreement, because the same height~$k+2$ is logged as both 0 and 1 in a single execution ($\sigma^x$).

Therefore, no protocol can achieve a tethered chain under partially synchronous primary chain writing when inactive nodes may behave arbitrarily.
\end{proof}

% tikz figure: 
\begin{figure}[t]
    \centering
    \resizebox{\columnwidth}{!}{
    \begin{tikzpicture}[
        timeline/.style={
            ->,
            thick,
            draw=black!70
        },
        label/.style={font=\Large},
        timepoint/.style={
            circle,
            fill=black,
            inner sep=1pt
        }
    ]
    
    % Define timeline width and spacing
    \def\verticalspacing{1.5}
    \def\timelinestart{0}
    \def\timelineend{9}
    
    % σ⁰ execution
    \begin{scope}[yshift=3*\verticalspacing cm]
        \node[label, scale=1.3] at (-0.5,0) {$\mathbf{\sigma^0}$};
        \draw[timeline] (\timelinestart,0) -- (\timelineend,0);
        \node[timepoint] at (0,0) {};
        \node[timepoint] at (2.5,0) {};
        \node[red, label] at (0,0.7) {$\gst + \progressStabilizationTime$};
        \draw[red, dashed] (0,0.5) -- (0,0);
        \node[label] at (2.5,0.5) {$C_{old}$ logs \textcolor{blue}{0}};
    \end{scope}
    
    % σ¹ execution
    \begin{scope}[yshift=2*\verticalspacing cm]
        \node[label, scale=1.3] at (-0.5,0) {$\mathbf{\sigma^1}$};
        \draw[timeline] (\timelinestart,0) -- (\timelineend,0);
        \node[timepoint] at (0,0) {};
        \node[timepoint] at (2.5,0) {};
        \node[red, label] at (0,0.7) {$\gst + \progressStabilizationTime$};
        \draw[red, dashed] (0,0.5) -- (0,0);
        \node[label] at (2.5,0.5) {$C_{old}$ logs \textcolor{green!60!black}{1}};
    \end{scope}
    
    % σ¹' execution
    \begin{scope}[yshift=\verticalspacing cm]
        \node[label, scale=1.3] at (-0.5,0) {$\mathbf{\sigma^{1'}}$};
        \draw[timeline] (\timelinestart,0) -- (\timelineend,0);
        \node[timepoint] at (0,0) {};
        \node[timepoint] at (2.5,0) {};
        \node[timepoint] at (5,0) {};
        \node[timepoint] at (7.5,0) {};
        \node[label] at (2.5,0.5) {$C_{old}$ logs \textcolor{green!60!black}{1}};
        \node[label] at (7.5,0.5) {$C_{new}$ logs \textcolor{green!60!black}{1}};
        \node[label] at (7.5,1.0) {\textit{(agreement)}};
        \node[red, label] at (5,0.7) {$\gst'$};
        \draw[red, dashed] (5,0.5) -- (5,0);
    \end{scope}
    
    % σˣ execution with timestamps
    \begin{scope}[yshift=0cm]
        \node[label, scale=1.3] at (-0.5,0) {$\mathbf{\sigma^x}$};
        \draw[timeline] (\timelinestart,0) -- (\timelineend,0);
        \node[timepoint] at (0,0) {};
        \node[timepoint] at (2.5,0) {};
        \node[timepoint] at (5,0) {};
        \node[timepoint] at (7.5,0) {};
        
        % Committee labels with underbraces under timeline (symmetric spacing and closer text)
        \draw [decoration={brace,mirror,raise=3pt},decorate] (0.2,-0.05) -- (2.3,-0.05) 
            node [below=5pt,pos=0.5] {$C_{old}$};
        \draw [decoration={brace,mirror,raise=3pt},decorate] (2.7,-0.05) -- (4.8,-0.05)
            node [below=5pt,pos=0.5] {$C_{old}$ unstaking};
        \draw [decoration={brace,mirror,raise=3pt},decorate] (5.2,-0.05) -- (7.3,-0.05)
            node [below=5pt,pos=0.5] {$C_{new}$};
        
        % Connect timestamps to circles with labels at the bottom
        \draw[dotted] (0,0) -- (0,-0.8);
        \node[label] at (0,-1) {$\mathbf{t^{\textit{divergence}}}$};
        \draw[dotted] (2.5,0) -- (2.5,-0.8);
        \node[label] at (2.5,-1) {$\mathbf{t^{\textit{decide}}}$};
        \draw[dotted] (5,0) -- (5,-0.8);
        \node[label] at (5,-1) {$\mathbf{t^{\textit{decide}} + \deltaActive}$};
        \draw[dotted] (7.5,0) -- (7.5,-0.8);
        \node[label] at (7.5,-1) {$\mathbf{t^{\textit{decide}'}}$};
        
        \node[label] at (2.5,0.5) {$C_{old}$ logs \textcolor{blue}{0}};
        \node[label] at (7.5,0.5) {$C_{new}$ logs \textcolor{green!60!black}{1}};
        \node[label] at (7.5,1.0) {\textit{(indistinguishability from $\sigma^{1'}$)}};
        \node[red, label] at (9.5,0.5) {$\gst^x$};
        \draw[red, dashed] (9.5,0.3) -- (\timelineend,0);
    \end{scope}
    
    \end{tikzpicture}
    }
    \Description{Illustration of impossibility proof.}
    \caption{Impossibility proof.}
    \label{fig:impossibility}
\end{figure}

%%%%%%%%%%%%%%%%%%%%%%%%%%%%%%%%%%%%%%%
%%%%%%%%%%%%%%%%%%%%%%%%%%%%%%%%%%%%%%%
%%%%%%%%%%%%%%%%%%%%%%%%%%%%%%%%%%%%%%%

\section{Aegis} \label{sec:aegis}

After establishing that no secure tethered chain can be achieved in a partially synchronous setting, we present Aegis, a tethered-chain protocol that takes advantage of synchronous access to a primary ledger.
We begin with an overview of Aegis~(\secQuickRef{sec:aegis:overview}).
Then we describe its components in detail: Aegis's blockchain structure~(\secQuickRef{sec:aegis:dataStructure}), its primary-chain contract~(\secQuickRef{sec:aegis:contract}), and the peer-to-peer protocol that Aegis nodes execute to form the blockchain~(\secQuickRef{sec:aegis:p2p}).
Finally, we discuss the execution of the forensics procedure to identify and penalize misbehaving nodes~(\secQuickRef{sec:aegis:forensics}).

%%%%%%%%%%%%%%%%%%%%%%%%%%%%%%%%%%%%%%%

\subsection{Protocol Overview} \label{sec:aegis:overview}

%%%%%%%%%%%%%%%%%%%%%%%%%%%%%%%%%%%%%%%

% Aegis comprises two components working in tandem: a distributed protocol executed by Aegis nodes and a smart contract deployed on the primary chain.
% Together they implement a secure tethered chain.

\paragraph{Aegis Blockchain.}
The Aegis blockchain advances by blocks generated by committees.
Committees use a Byzantine fault-tolerant consensus protocol designed for partial synchrony.
Figure~\ref{fig:aegisOperation} illustrates Aegis's blockchain topology and its key mechanisms.
Each Aegis block contains hash references to its parent block and to a primary-chain block.
The parent reference establishes the chain's history, while the primary-chain reference defines the committee responsible for producing the next block.
For instance, Aegis block $A1$ references primary block $P2$, so $P2$'s state determines the committee for block $A2$.
Similarly, $A2$'s reference to $P2$ means $P2$'s state determines the committee for block $A3$.

\begin{figure}[t]
    \centering
    \resizebox{\columnwidth}{!}{
        \begin{tikzpicture}[
    every node/.style={rectangle, draw,minimum width=0.3cm, minimum height=.3cm},
    >=Stealth,
    bend angle=45, % Control the bend of curved arrows
    scale=1.0, % Control the bend of curved arrows
    transform shape=true,
]

\node[above, draw=none, anchor=east] at (9.75,2.5) {Aegis chain};
\node[above, draw=none, anchor=east] at (9.75,-0.5) {Primary chain};

% Aegis chain nodes
\node[draw=none, fill=none] (A0) at (-0.5,2) {...};
\node[fill=magenta!30] (A1) at (1.75,2) {$A1$};
\node[fill=green!30] (A2) at (3,2) {$A2$};
\node[fill=green!30] (A3) at (4.25,2) {$A3$};
\node[fill=cyan!30] (A4) at (8.75,2) {$A4$};
\node[draw=none, fill=none] (Ar) at (9.5,2) {...};
\node[draw=none, fill=none] (Adel) at (6.5,2) {$//$};

% Primary chain nodes
\node[draw=none, fill=none] (P0) at (-0.5,0) {...};
\node[fill=magenta!30] (P1) at (0.5,0) {$P1$};
\node[fill=green!30] (P2) at (1.5,0) {$P2$};
\node[fill=white!30] (P3) at (4.75,0) {$P3$};
\node[fill=cyan!30] (P4) at (7.5,0) {$P4$};
\node[fill=white!30] (P5) at (8.5,0) {$P5$};
\node[draw=none, fill=none] (Pr) at (9.5,0) {...};
\node[draw=none, fill=none] (Pdel) at (6.5,0) {$//$};

\node[draw=none, fill=none] (period1delta) at (1.35,-1) {\Large $\deltaActive$};
\node[draw=none, fill=none, anchor=east] (period1start) at (0.5,-1) {Com 1};
\node[draw=none, fill=none] (period1end) at (5.15,-1) {};

\node[draw=none, fill=none] (period2delta) at (2.6,-1.5) {\Large $\deltaActive$};
\node[draw=none, fill=none, anchor=east] (period2start) at (1.5,-1.5) {Com 2};
\node[draw=none, fill=none] (period2end) at (6.4,-1.5) {};

\node[draw=none, fill=none] (period3delta) at (5.5,-2) {\Large $\deltaActive$};
\node[draw=none, fill=none, anchor=east] (period3start) at (4.75,-2) {No resets};
\node[draw=none, fill=none] (period3end) at (7.25,-2) {};
\node[draw=none, fill=none] (period3jump) at (6.5,-2) {$//$};

\node[draw=none, fill=none] (period4delta) at (8.35,-1) {\Large $\deltaActive$};
\node[draw=none, fill=none, anchor=east] (period4start) at (7.5,-1) {Com 4};
\node[draw=none, fill=none] (period4end) at (9.5,-1) {...};

% Draw arrows (modify as needed)
\draw[->, line width=0.8pt, dashed] (A1) -- (P1) node[midway, left, draw=none, scale=0.8, font=\Large, yshift=0.1cm] {reset};
\draw[->, line width=0.8pt,] (A1) -- (P2) node[midway, right, draw=none, scale=0.8, font=\Large, yshift=0.1cm] {ref};
\draw[->, line width=0.8pt,] (A2) -- (P2) node[midway, right, draw=none, scale=0.8, font=\Large, xshift=0.1cm, yshift=0.1cm] {ref};
\draw[->, line width=0.8pt,] (A3) -- (P2) node[midway, right, draw=none, scale=0.8, font=\Large, xshift=0.2cm, yshift=0.1cm] {ref};

\draw[->, line width=0.8pt, dashed] (A4) -- (P4) node[midway, left, draw=none, scale=0.8, font=\Large, yshift=0.1cm] {reset};
\draw[->, line width=0.8pt,] (A4) -- (P5) node[midway, right, draw=none, scale=0.8, font=\Large, yshift=0.1cm] {ref};

\draw[->,line width=0.8pt,bend right=20] (P3) to node[midway, right, draw=none, scale=0.8,font=\Large, xshift=0.1cm, yshift=0.1cm] {chk} (A3);

% Connect upper chain nodes with arrows
\draw[<-] (A1) edge (A2) 
          (A2) edge (A3) 
          (A3) edge (A4) 
          (A4) edge (Adel)
          (A4) edge (Ar)
          (A0) edge (A1);

% Connect lower chain nodes with arrows
\draw[<-] (P1) edge (P2) 
          (P2) edge (P3)
          (P3) edge (Pdel)
          (P4) edge (P5)
          (P5) edge (Pr)
          (P0) edge (P1);

\draw[-|]{}
(period1delta) edge  (period1start)
(period1delta) edge  (period1end)
(period2delta) edge  (period2start)
(period2delta) edge  (period2end)
(period3jump) edge  (period3end)
(period3delta) edge  (period3start)
(period4delta) edge  (period4start);

\draw[-]{}
(Pdel) edge (P4)
(Adel) edge (A4)
(period3delta) edge  (period3jump)
(period4delta) edge  (period4end);

\end{tikzpicture}
    }
    \caption{Aegis's toplogy and key mechanisms.}
    \label{fig:aegisOperation}
\end{figure}

\paragraph{Committee Selection.}
A primary block defines a committee from nodes with locked stake weighted by their stake amount.
\begin{note}
    Committees can be selected in various ways to utilize the consensus protocol more efficiently.
    For example, committees can include only nodes exceeding a minimum stake threshold or a fixed number of highest-stake nodes.
\end{note}
\begin{note}
    In staking transactions, nodes must include contact information (e.g., IP address) for message routing; we omit this detail for simplicity.
\end{note}

Aegis handles committee transitions without explicit handovers through references to primary blocks.
As shown in Figure~\ref{fig:aegisOperation}, committee 1 initializes committee 2 as block~$A2$ references primary block~$P2$.
Each committee remains active for a bounded period~($\deltaActive$) after the referenced primary block.
During this period, members cannot withdraw their stake.
This provides the accountability necessary for \emph{safety}, meaning that the chain is agreed upon by all correct nodes and cannot be altered retroactively.

\paragraph{Primary-Chain Smart Contract.}
In addition to primary block references, Aegis employs two complementary mechanisms through its primary-chain smart contract: checkpoints and resets.

The checkpoint mechanism anchors the Aegis state to the primary chain, creating tamper-proof reference points that (1)~help new nodes determine the correct chain history and (2)~signal the smart contract that the Aegis chain is currently progressing.
Active committees submit a checkpoint before the end of their activity period.
% Figure info
For instance, a node in committee 1 submits a checkpoint of block~$A3$ to~$P3$.
The smart contract verifies that the block's committee is active during submission, that the block results from valid consensus, and that it is a descendant of the last checkpointed block.

The reset mechanism helps Aegis recover from communication failures.
If network delays prevent a committee from producing blocks before its activity period expires, no new committee will be selected.
In this case, any node can issue a reset through the primary-chain contract.
A reset establishes a new committee based on the current primary-chain state, which can then resume extending the Aegis chain from the last checkpointed block.
The committee produces a block that includes a hash reference to the primary-chain block containing the reset.
% Figure info
For example, after committee 2's activity period ended, a reset was submitted to block~$P4$, creating a committee that produced a block referencing~$P4$.

\paragraph{Node Protocol.}
Aegis nodes execute a distributed protocol that utilizes the primary chain, the smart contract, and its checkpoint and reset mechanisms.
Each node maintains a local view of the Aegis blockchain, coordinated with other nodes via peer-to-peer communication.

Nodes follow the latest state of the primary chain: they reference the latest block when producing an Aegis block and monitor the smart contract for checkpoints and resets.
Nodes have access to previous primary-chain blocks when needed and can post checkpoints and resets to the contract.
Aegis relies on these actions taking at most~$\deltaPrimaryWrite$ time steps.

The protocol operates in three sequential phases at each time step.
First, nodes read the primary chain to identify the latest checkpointed block and determine the active committee.
Then, they log all Aegis blocks created by active committees since the last checkpoint.
Finally, nodes run a Byzantine fault tolerant consensus protocol to extend the Aegis chain.

\paragraph{Timing Constraints.}
There are two challenges that might arise when a committee approaches the end of its activity period.
First, a committee might benignly fail to checkpoint a block or forward it to the next committee.
This would result in the next committee creating a conflicting block that does not reference the previous committee's block, violating agreement.
Second, committee nodes might maliciously produce conflicting blocks right before their period ends, leaving insufficient time for forensics and penalties, removing incentives to behave properly.

Aegis addresses these challenges with strict timing constraints to ensure sufficient time for checkpoints and forensics.
Near the end of the activity period, nodes stop producing new blocks.
Instead, they checkpoint their last generated block to ensure proper handover, then wait for their tenure to end.
This period covers the last~$3\deltaPrimaryWrite$ steps of each committee's activity period.

Another challenge is premature resets that could create concurrent committees and conflicting blocks.
To prevent this, the smart contract accepts resets only when the latest checkpoint or reset is at least~$\deltaActive$ time steps old.
% Figure info
For example, since~$P1$ contains a reset, no other reset is accepted during committee 1's activity period, lasting~$\deltaActive$ steps from~$P1$'s creation.
Additionally, a checkpoint to~$P3$ restarts this restriction for~$\deltaActive$ steps following the checkpoint.

Next, we describe the components of Aegis in detail, starting with the data structure of the Aegis blockchain.

%%%%%%%%%%%%%%%%%%%%%%%%%%%%%%%%%%%%%%%
%%%%%%%%%%%%%%%%%%%%%%%%%%%%%%%%%%%%%%%
%%%%%%%%%%%%%%%%%%%%%%%%%%%%%%%%%%%%%%%

    \subsection{Chain Structure} \label{sec:aegis:dataStructure}

Aegis's data structure includes the Aegis chain and relevant updates in the primary chain. 

All nodes start with a single Genesis Aegis block.
All Aegis blocks~$b$ (apart from Genesis) have a parent Aegis block~$\parent{b}$ and an Aegis-to-Primary reference to a primary block~$\LOneRef{b}$.
(For genesis only:~$\parent{\bGenesis} = \bot$ and~$\LOneRef{\bGenesis} = \bot$.)

Each Aegis block~$b$ should be generated by consensus among the committee it specifies. 
The committee is specified in one of two ways. 
First, an aegis block can have a reference to a reset in a primary block, denoted~$\LOneResetRef{b}$. 
If~$\LOneResetRef{b} \neq \bot$, the committee is the one specified by the reset.
% For a primary block~$B$, we denote by~$\blockMembership{B}$ the set of nodes with stake at block~$B$. 
Otherwise, the committee is the one specified in the primary block referenced by~$b$'s parent, i.e., $\LOneRef{\parent{b}}$. 
The function \blockMembership{B} returns the committee specified by a primary block~$B$. 
The block should thus be valid according to the \consensusValidate function with the appropriate committee. 
The consensus instance ID is the pair consisting of its parent and reset references. 

The function \isValid (Algorithm~\ref{alg:isValid}) validates all this---Note that this is only a data structure validation that does not take into account committee activity. 
The protocol ignores blocks that are invalid as checked by $\isValid$. 

% IsValid: 
\begin{algorithm}[t] 
    \SetAlgoNoLine
    \SetAlgoNoEnd 
    \DontPrintSemicolon 
    \caption{Valid block predicate} 
    \label{alg:isValid}
    \fontsize{7.5}{8.5}\selectfont 
    \KwFunction({\isValid($b, \primaryLedger, \mathcal{B}$)}){ 
        \lIf{$b$ is Genesis}{return \vTrue}
        % \lIf{A checkpoint in \primaryLedger conflicts with $b$}{return $\vFalse$} % IE: Remove? 
        \lIf{$\parent{b} \not\in \mathcal{B}$}{return $\vFalse$}
        \lIf{$\LOneRef{b} \not\in \primaryLedger$}{return $\vFalse$} 
        \lIf{$\LOneRef{b}$ is not a descendant of~$\LOneRef{\parent{b}}$}{return $\vFalse$} 
        \lIf(\algoComment{Recursively validate parents}){$\lnot \isValid(\parent{b}, \primaryLedger, \mathcal{B})$}{return $\vFalse$} 
        \If{$\LOneResetRef{b} \neq \bot$}{
            \lIf{$\LOneResetRef{b} \not\in \primaryLedger$}{return $\vFalse$}
            \lIf{$\LOneResetRef{b}$ is not a descendant of~$\LOneRef{\parent{b}}$}{return $\vFalse$} \label{alg:isValid:resetAndParentRefOrder}
            $\committee \gets \blockMembership{ \LOneResetRef{b}$ } \; 
        } \Else {
            $\committee \gets \blockMembership{ \LOneRef{b}$ } \; 
        }
        \lIf{$\lnot\consensusValidate[b][\bigl( \blockParent{b}, \LOneResetRef{b} \bigr)][\committee]$}{return $\vFalse$} \label{alg:isValid:consensusValidate}
        return \vTrue
    }
\end{algorithm}

%%%%%%%%%%%%%%%%%%%%%%%%%%%%%%%%%%%%%%%%%%%%%%%%%%%%%%%%%%%
%%%%%%%%%%%%%%%%%%%%%%%%%%%%%%%%%%%%%%%%%%%%%%%%%%%%%%%%%%%

% Aegis contract: 
\begin{algorithm}[t] 
    \SetAlgoNoLine 
    \SetAlgoNoEnd 
    \DontPrintSemicolon 
    \caption{Aegis primary-ledger contract.} 
    \label{alg:contract} 
    \fontsize{7.5}{8.5}\selectfont 
    \KwOn(Receive at time~$t$ entry~$e$ referencing Aegis block~$b$ and its parent~$b'$){
        \If{$e$ is a reset}{
            assert no entries in blocks later than~$t - \deltaActive$ \label{alg:contract:resetOnlyAfterStale} {\commentColor\KwIttayComment*{First entry or previous stale}}  
        } \Else(\algoComment{checkpoint}) {
            assert $\LTwoRef{e} = b$ \;
            assert~$\blockParent{b} = b'$ \; 
            assert $\LOneRef{b} \in \primaryLedger$ \;  
            \If(\algoComment{$b$'s committee specified by reset}){$\LOneResetRef{b} \neq \bot$} {
                assert~$\LOneResetRef{b} \in \primaryLedger$ \; 
                $\EReset \gets$ primary block containing reset entry at $\LOneResetRef{b}$ \label{alg:contract:checkpointAfterReset} {\commentColor\KwIttayComment*{Primary-chain block with reset entry}}
                $\mathcal{N} \gets \blockMembership{\EReset}$ \; 
                $t_0 \gets \blockTime{\EReset}$ \;
            } \Else(\algoComment{$b$'s committee specified by previous block}) {
                $\EReset \gets \bot$ \; 
                assert~$\LOneRef{b'} \in \primaryLedger$ {\commentColor\KwIttayComment*{Parent's reference exists}}
                $\mathcal{N} \gets$ \blockMembership{\primaryLedger(\LOneRef{b'})} {\commentColor\KwIttayComment*{Specified by parent}}
                $t_0 \gets \blockTime{\primaryLedger(\LOneRef{b'})}$ \; 
                % \ie{ \If{previous entry is a reset} {
                %     assert~$\blockHeight{b}$ is greater by 2 from the last checkpointed block \label{alg:contract:checkpointAfterReset} \; 
                % }}
            }
            assert~$t < t_0 + \deltaActive$ \label{alg:contract:timeValidate} \; 
            assert~$\consensusValidate[b][\bigl( b', \EReset \bigr)][\committee]$ \label{alg:contract:consensusValidate} \; 
            % assert $\blockHeight{b} > \blockHeight{b'}$ \; 
            assert $\blockHeight{b} > \blockHeight{\LTwoRef{\textup{previous checkpoint}}}$ \;
        }
        accept~$e$ {\commentColor\KwIttayComment*{All asserts passed, register checkpoint/reset}}
    }
    \end{algorithm}

Both Aegis blocks and primary blocks have associated generation times.
For any block~$b$, we denote its generation time by~$\blockTime{b}$, representing the timestamp when the block was created.
In the protocol, only the timestamps of primary blocks are used to determine committee activity periods and to enforce timing constraints to avoid potential vulnerabilities via timestamp manipulation of Aegis blocks.

        \subsection{Primary-Chain Contract} \label{sec:aegis:contract}

The primary-chain contract (Algorithm~\ref{alg:contract}) allows checkpointing Aegis blocks and initiating resets. 
It receives \emph{entries} submitted by Aegis nodes asking to register either a checkpoint or a reset.
In each block the primary-chain contract may receive one entry at most, and ignores any further submissions.

The contract registers a reset if there are no entries in the last~$\deltaActive$ steps, indicating one is indeed necessary; no other checks are required. 

For a checkpoint, the contract receives an Aegis block to checkpoint along with its parent. 
It verifies that the block correctly references a primary block earlier than this checkpoint. 
Similarly, if the Aegis block has a reset reference, the contract verifies it points to a previous primary block. 
It also verifies that the block's Aegis parent references a previous primary block. 
The contract then verifies that the committee that created the block is still active by comparing the current time (known to the contract) to the time of the committee's referenced block. 
It validates the block is the result of the committee's consensus by using the function \consensusValidate.
Finally, it verifies that the block is higher (farther from the Aegis Genesis) than the last checkpointed block.
Having passed all these checks, the contract accepts the checkpoint and registers it.

\begin{note}
For a practical implementation, to reduce primary-chain overhead, the checkpoint can include only the hash of the Aegis block~$b$, and the following elements with proof they are included in the block: reference from~$b$ to the primary block, reference from $b$ to a reset (could be~$\bot$), reference from~$b$ to its parent~$b'$. 
Similarly, for~$b'$ we have its reference to a primary block. 
\end{note}

\begin{note}
Theoretically, neglecting slashing enforcement, a protocol-agnostic append-only log would suffice instead of a smart contract, with the nodes reading it and locally applying the contract logic. 
\end{note}

%%%%%%%%%%%%%%%%%%%%%%%%%%%%%%%%%%%%%%%%%%%%%%%%%%%%%%%%%%%
%%%%%%%%%%%%%%%%%%%%%%%%%%%%%%%%%%%%%%%%%%%%%%%%%%%%%%%%%%%

% Generate a label for the Algorithm number without a/b for parts: 
\newcounter{AlgoAegisCounter} % New counter for the 2-part algorithm
\setcounter{AlgoAegisCounter}{\value{algocf}} % Set to the first part's number
\refstepcounter{AlgoAegisCounter}\label{alg:aegis} % Label for the number (No a/b)

\renewcommand{\thealgocf}{\arabic{algocf}a} % First part of the algorithm (a)
\newcounter{savedlineno} % Create a new counter to store the line number

% Aegis node:
\begin{algorithm}[t]
    \SetAlgoNoLine 
    \SetAlgoNoEnd 
    \DontPrintSemicolon 
    \caption{Aegis algorithm for node~$i$. (part 1/2)} 
    \label{alg:aegisA} 
    \fontsize{7.5}{8.5}\selectfont  
    \nonl Initialization:  \;
    $\blockSet_i$: A mapping, initially $\blockSet_i(\blockID{\bGenesis}) = \bGenesis$, $\forall \textit{id} \neq \blockID{\bGenesis}: \blockSet_i(\textit{id}) = \bot$. \;
    $\ledger_i$: A vector, initially $\ledger_i(0) = \bGenesis$, $\forall k > 0: \ledger_i(k) = \bot$ \label{alg:aegis:logGenesis} \;
    \BlankLine
    \BlankLine
    \nonl Execution at step~$t$: (return immediately if a referenced block is missing) \;
    $\blockSet_i \gets \blockSet_i \cup \textup{blocks received}$ \label{alg:aegis:readStatesFirst} \;
    $\primaryLedger_i \gets \primaryLedger_i \| \textup{Primary chain extension receieved}$ \label{alg:aegis:readStatesLast} \algoComment{No conflicts by assumption} 
    \BlankLine
    \BlankLine
    $\ELast \gets $ primary block containing the last contract entry in $\primaryLedger_i$ or $\bot$ if none \; 
    \If{entry in \ELast is a reset \label{alg:aegis:isResetValid}}{
        \If{$t < \blockTime{\ELast} + \deltaActive - \deltaPrimaryWrite$ \label{alg:aegis:resetEnoughTimeCheck}}{
            $\BCheckpoint \gets $ primary block containing the latest checkpoint entry, $\bot$ if none \label{alg:aegis:checkpointBeforeReset}\label{alg:aegis:lastIsResetFirst} \;
            $\committee \gets \blockMembership{\ELast}$ \label{alg:aegis:resetCommittee} \; 
            $t_0 \gets \blockTime{\ELast}$ \label{alg:aegis:noteResetTime} \label{alg:aegis:lastIsResetLast} \; %{\color{gray}\KwIttayComment*{Timer starts at checkpoint or reset}}
        } 
        \ElseIf(\algoComment{Previous reset stale}){$t > \blockTime{\ELast} + \deltaActive$}{ 
            issue reset and return \label{alg:aegis:resetAfterReset}
        }
        \Else(\algoComment{Too close to timeout}) {
            return
        }
    } \ElseIf{entry in \ELast is a checkpoint \label{alg:aegis:lastIsCheckpointFirst}}  {
        $\BCheckpoint \gets \ELast$ \;
        $\bCheckpointed \gets \blockSet_i(\LTwoRef{\BCheckpoint})$ \; {\commentColor \KwIttayComment*{The checkpointed Aegis block}}
        $\committee \gets \blockMembership{ \primaryLedger (\LOneRef{\bCheckpointed}) }$ \label{alg:aegis:committeeFromCheckpoint} {\commentColor \KwIttayComment*{Committee from checkpointed Aegis block}}
        $t_0 \gets \blockTime{ \primaryLedger (\LOneRef{\bCheckpointed}) }$ \label{alg:aegis:noteCheckpointTime} \label{alg:aegis:lastIsCheckpointLast} {\commentColor \KwIttayComment*{Its generation time}}
    } \Else(\algoComment{No block found}) { 
        issue reset and return {\commentColor \KwIttayComment*{Initialize}}
    }

% At this point: 
%     Bcheckpoint: latest checkpoint 
%     N: Committee for next block after latest L1-logged committee, either from checkpoint or reset
%     t_0: Start time of latest L1-logged committee 

\If(\algoComment{First non-genesis block, committee by reset}\label{alg:aegis:blockOneFirst}){$\BCheckpoint = \bot$}{
    $k \gets 1$ \; 
    $\bCheckpointed \gets \bGenesis$ \; 
    $b \gets \bGenesis$ \label{alg:aegis:blockOneLast} \; 
} \Else {
    $\bCheckpointed \gets \blockSet_i(\LTwoRef{\BCheckpoint})$  {\commentColor\KwIttayComment*{We have now found the last checkpoint}}
    % \lIf(\algoComment{Perhaps we don't have its parents yet}){$\lnot \isValid(\bCheckpointed, \committee, \bigl( \blockParent{\bCheckpointed}, \LOneResetRef{\bCheckpointed} \bigr), \blockSet_i, \primaryLedger)$}{return} 
    \ForAll(\algoComment{Log all blocks up to checkpoint, inclusive}){$k$ from~$|\ledger| + 1$ to~$\blockHeight{\bCheckpointed}$ \label{alg:aegis:logToCheckpointedFirst}} {
        \If{$\exists b \in \blockSet_i$ s.t.\ $b$ is an ancestor of \bCheckpointed and $\blockHeight{b} = k$}{
            $\ledger_i(k) \gets b$ \label{alg:aegis:logFromCheckpoint} \label{alg:aegis:logToCheckpointedLast} \algoComment{Already verified via contract}
            % $\ledger_i(k) \gets \blockSet_i(x)$ \label{alg:aegis:logFromCheckpoint} \label{alg:aegis:logToCheckpointedLast} \;
        } \Else {
            return \algoComment{We are missing an ancestor of the checkpoint}
        }
    } 
}

\setcounter{savedlineno}{\value{AlgoLine}} % Save the last line number
\end{algorithm}

\addtocounter{algocf}{-1} % Keep the algorithm counter (3), no progress
\renewcommand{\thealgocf}{\arabic{algocf}b} % Second part of the algorithm (b)

\begin{algorithm}[t]
\SetAlgoNoLine 
\SetAlgoNoEnd 
\DontPrintSemicolon 
\caption{Aegis algorithm for node~$i$. (part 2/2)}
\label{alg:aegisB}
\fontsize{7.5}{8.5}\selectfont  
\setcounter{AlgoLine}{\value{savedlineno}} % Set the starting line number
    
% \BlankLine 
$k \gets |\ledger|$ {\commentColor\KwIttayComment*{Equal to $\blockHeight{\bCheckpointed}$}} 
$b \gets \bCheckpointed$ \; 
\If(\algoComment{Committee for next block still active}){$t < t_0 + \deltaActive$ \label{alg:aegis:isPrevBlocksCommitteeActive}} {
    Choose $b' \in \blockSet_i$ s.t.\ $\isValid(b', \mathcal{N}, \primaryLedger, \blockSet_i)$ and $\blockParent{b'} = \bCheckpointed$, $\bot$ if none \label{alg:aegis:chooseBprimeBeforeLoop} \; 
    \While(\algoComment{Add blocks by active committees}){
    $b' \neq \bot$
    \label{alg:aegis:committeeActive}}{
        $b \gets b'$ \; 
        $\ledger(k) \gets b$ \label{alg:aegis:logByPrevConsensus} \;  
        $k \gets k + 1$ \; 
        $\mathcal{N} \gets \blockMembership{\LOneRef{b}}$ \label{alg:aegis:changeCommittee} \; 
        % \mbox{\color{green!70!black} \sout{$t_0 \gets \blockTime{\LOneRef{b}}$}}  \;
        Choose $b' \in \blockSet_i$ s.t.\ $\isValid(b', \mathcal{N}, \primaryLedger_i, \blockSet_i)$ and $\blockParent{b'} = b$, $\bot$ in none \label{alg:aegis:chooseBprimeInLoop} \;
    }
}
\BlankLine
\If(\algoComment{Latest committee has enough time to extend}){$t < t_0 + \deltaActive - 3 \deltaPrimaryWrite$ \label{alg:aegis:enoughTimeToExtend}}{
    \If(\algoComment{No blocks since last checkpoint and following a reset}){$b = \bCheckpointed \land$ entry in $\ELast$ is a reset} {
        $b \gets \consensusStep[\inputFuncITK{i}{t}{|\ledger|}][\bigl( \ledger(|\ledger|-1), \ELast \bigr)][\mathcal{N}]$ \label{alg:aegis:runConsensusAfterReset} \; 
    } \Else(\algoComment{Last checkpointed or subsequent block's committee is active}) {
        $b \gets \consensusStep[\inputFuncITK{i}{t}{|\ledger|}][\bigl( \ledger(|\ledger|-1), \bot \bigr)][\mathcal{N}]$ \label{alg:aegis:runConsensus} \; 
    } 
    \lIf{$b \neq \bot$}{$\ledger(|\ledger|) \gets b$} \label{alg:aegis:logByCurrConsensus} 
}
\BlankLine
\If(\algoComment{Deadline to issue checkpoint}){$t = \blockTime{\ELast} + \deltaActive - 3 \deltaPrimaryWrite$ and $b \neq \bCheckpointed$
\label{alg:aegis:checkpointLastMinuteFirst}}{
    issue checkpoint for~$b$ and return \label{alg:aegis:checkpointLastMinute} \;
}
\If{$t > t_0 + \deltaActive$} {
    issue reset and return \label{alg:aegis:resetAfterCheckpoint} \;
}
{\commentColor\KwIttayComment*{No-op if $t_0 + \deltaActive - 3\deltaPrimaryWrite \leq t \leq t_0 + \deltaActive$}} \label{alg:aegis:giveUpTooLate} 
\end{algorithm}

% After the split algorithms
% \setcounter{algocf}{\value{tempAlgoCounter}} % restore the counter
\renewcommand{\thealgocf}{\arabic{algocf}}

        \subsection{Node Protocol} \label{sec:aegis:p2p}

We are now ready to present the protocol executed by each node (Algorithm~\ref{alg:aegis}). 
All nodes start with a single Aegis Genesis block (line~\ref{alg:aegis:logGenesis}) and in each time step proceed in three phases, as follows. 

%%%%%%%%%%%%%%%%%%%%%%%%%%%%%%%%%%%%%%%%%%%%%%%%%%%%%%%%%%%
            
            \subsubsection{Sync to Latest Checkpoint}

At the beginning of a step a node first reads the states of the primary and tethered blockchains (lines~\ref{alg:aegis:readStatesFirst}--\ref{alg:aegis:readStatesLast}). 
Then, it finds the latest checkpointed block as follows. 
If the latest primary-chain entry is a checkpoint, it takes this checkpoint; the committee for the next block to log is expected to be the one defined by the checkpointed block, i.e., according to the primary block it references (line~\ref{alg:aegis:committeeFromCheckpoint}). 
Otherwise, if the latest entry is a reset and the committee defined by the reset is still active (line~\ref{alg:aegis:isResetValid}), it takes the most recent checkpoint with the reset's committee for the subsequent block. 
If the reset's committee is inactive, the node issues a new reset and returns, waiting for it to take place. 
Finally, if there is no primary block entry then this is the first execution of the protocol, so the node issues a reset and returns, waiting for it to take place. 

If the only entry was a reset, namely, there are no checkpointed blocks, then this is the first block after Genesis (lines~\ref{alg:aegis:blockOneFirst}--\ref{alg:aegis:blockOneLast}). 
Otherwise, the node goes from the latest block it had confirmed in previous steps (maybe starting from Genesis), logging in its local Aegis chain all ancestors up to the checkpointed block (lines~\ref{alg:aegis:logToCheckpointedFirst}--\ref{alg:aegis:logToCheckpointedLast}). 
If the node is yet to receive any of those, the algorithm stops; the node will wait to receive them in subsequent steps. 

% In each primary block~$B$, there might be an Aegis contract entry with a reference to an Aegis block (See Section~\ref{sec:aegis:contract}). 
% Denote by~\LTwoRef{B} the Aegis block (in~$\ledger$) referenced by the contract call in the primary-ledger block~$B$ (if there is such a call). 
% % Note that blocks hold hash references to other blocks at the same or different chains. 
% Denote by~$\hash(\cdot)$ the function calculating the hash of a block and by~$\hat{B}$ the hash of block~$B$. 
% Therefore, for any two blocks~$B$ and~$B'$ the adversary finds, ${\hash(B') = \hat{B}}$ if and only if~${B = B'}$, except with negligible probability. 

            \subsubsection{Sync Past Checkpoint} 

Having caught up to the latest checkpoint, the node proceeds to log Aegis blocks finalized by consensus that are yet to be checkpointed. 
This is only done if the latest checkpointed block's committee is still active (line~\ref{alg:aegis:isPrevBlocksCommitteeActive}), implying all subsequent committees are still active. 

For each block, it adds it to the ledger and updates the next block's committee (lines~\ref{alg:aegis:isPrevBlocksCommitteeActive}--\ref{alg:aegis:chooseBprimeInLoop}). 
This is done for all Aegis blocks that were generated with currently active committees. 
Recall that the committee of the first block after the checkpoint is defined by the primary block its parent points to (lines~\ref{alg:aegis:lastIsCheckpointFirst}--\ref{alg:aegis:lastIsCheckpointLast}), or by a reset block if it was produced that way (lines~\ref{alg:aegis:lastIsResetFirst}--\ref{alg:aegis:lastIsResetLast}). 

Once this is done, the node is up-to-date, having logged all previous blocks before the last checkpoint and after the checkpoint (with active committees). 

            \subsubsection{Extend the log} 

If the latest Aegis block specified committee has enough time remaining active, the node tries to extend the chain by taking a consensus protocol step. 
Specifically, there should be enough time to checkpoint a new block and complete the forensics procedure (line~\ref{alg:aegis:enoughTimeToExtend}). 
The consensus is identified by the parent and reset blocks (line~\ref{alg:aegis:runConsensusAfterReset} or~\ref{alg:aegis:runConsensus}). 
If the committee is no longer active, it issues a reset (line~\ref{alg:aegis:resetAfterCheckpoint}) 
to enable progress in future steps. 
Because the committee defined by the checkpoint is inactive, the contract will accept the reset. 

Aegis nodes issue a checkpoint before the committee expires~(line~\ref{alg:aegis:checkpointLastMinute}), which is sufficient for correctness. 
Note that it could be that just one correct node got the block (line~\ref{alg:aegis:runConsensusAfterReset} or~\ref{alg:aegis:runConsensus}) before the committee expires---the consensus protocol does not guarantee more. 
Then the perpetuation of this block relies entirely on this node issuing the checkpoint. 
But this is guaranteed since this node is correct. 

\begin{note}
The consensus ID for each Aegis block~$b$ includes both its ancestor~$\LOneRef{b}$ and its reset reference~$\LOneResetRef{b}$, the latter is~$\bot$ if the block's committee is specified by its ancestor. 
This guarantees that the committee could only start the consensus instance when it was active. 
Without the reset reference, a committee $C$ could be specified by a primary block, later become stale, create a block~$b$ while being inactive, then a reset reinstates the same committee, which is now active, allowing it to present~$b$ as the result of an active~committee. 
\end{note}

%%%%%%%%%%%%%%%%%%%%%%%%%%%%%%%%%
%%%%%%%%%%%%%%%%%%%%%%%%%%%%%%%%%

        \subsection{Forensics} \label{sec:aegis:forensics}

Aegis takes advantage of the consensus protocol's forensics support~\cite{sheng2021bft,civit2021polygraph} to penalize Byzantine nodes if their number is greater than its threshold, and they successfully violate agreement. 
The forensics procedure encompasses both the contract and the peer-to-peer protocol, we describe it only here for pseudocode readability. 

A correct node identifies a violation if it observes two conflicting decided values (by a consensus step, receiving a block, or a checkpoint). 
In this case, it issues a transaction notifying the contract, which initiates the consensus protocol's forensics procedure. 
Once the other correct nodes observe the notification in the primary chain, they participate in the procedure, sending the necessary data to the contract. 
The contract then identifies the Byzantine nodes and penalizes them by slashing (revoking) their stake. 
The forensics procedure depends on the consensus algorithm's details, and optimization for reducing the primary-chain overhead is outside our scope. 

%%%%%%%%%%%%%%%%%%%%%%%%%%%%%%%%%%%%%%%%%%%%%%%%%
%%%%%%%%%%%%%%%%%%%%%%%%%%%%%%%%%%%%%%%%%%%%%%%%%
%%%%%%%%%%%%%%%%%%%%%%%%%%%%%%%%%%%%%%%%%%%%%%%%%

    \section{Security} \label{sec:security} 

%%%%%%%%%%%%%%%%%%%%%%%%%%%%%%%%%%%%%%%%%%%%%%%%%
%%%%% IF NOT ENOUGH ROOM, DEFER TO APPENDIX %%%%%
% The proof is deferred to Appendix~\ref{app:security}; we briefly review it here. 
% To prove security we must show that Aegis nodes do not register different values at the same height (Agreement), that if all have the same value and none is Byzantine they decide that value (Validity), that after~$\gst$ they register values frequently (Progress), and that violation leads to penalties. 

% We first show by induction that the ancestors of a block generated by an active committee were also generated by active committees. 
% Then, using an indistinguishability argument we show that the guarantees of an existing consensus protocol (with an infinitely-running committee) are preserved during the validity period of an Aegis committee. 

% The proofs of the three properties follow a similar structure. 
% We prove by induction on the Aegis block number. 
% For each we show the property holds either due to the consensus guarantees, or due to the interaction of the nodes protocol and the primary-chain contract. 
%%%%%%%%%%%%%%%%%%%%%%%%%%%%%%%%%%%%%%%%%%%%%%%%%
%%%%%%%%%%%%%%%%%%%%%%%%%%%%%%%%%%%%%%%%%%%%%%%%%

To prove Aegis's security, we first show that the ancestors of a block generated by an active committee were also generated by active committees~(\secQuickRef{sec:security:activeCommittee}).
After that, we prove Aegis achieves agreement~(\secQuickRef{sec:security:agreement}), validity~(\secQuickRef{sec:security:validity}) and progress~(\secQuickRef{sec:security:progress}), and that it penalizes misbehaving nodes~(\secQuickRef{sec:security:slashing}). 
The security guarantees follow~(\secQuickRef{sec:security:conclusion}). 
Appendix~\ref{app:securityAssumptions} summarizes the assumptions Aegis requires.

%%%%%%%%%%%%%%%%%%%%%%%%%%%%%%%%%%%%%%%%%%%%%%%%%
%%%%%%%%%%%%%%%%%%%%%%%%%%%%%%%%%%%%%%%%%%%%%%%%%
%%%%%%%%%%%%%%%%%%%%%%%%%%%%%%%%%%%%%%%%%%%%%%%%%

        \subsection{Active Committee Follows Active Committees} \label{sec:security:activeCommittee}

We call a committee active if its nodes' collateral was locked recently enough, so it could not be withdrawn yet. 
        % \ie{we first define the notion of an \emph{active committee}, whose members are active, that is, have collateral (stake) on the primary chain. }

\begin{definition}[Active committee]
A weighted set of nodes~$\committee$ is an \emph{active committee at time~$t$} if there exists a primary-chain block published in the interval~$[t - \deltaActive, t]$ and those nodes staked collateral up to that block and did not unstake by that block. 
\end{definition}

% \begin{definition}[Valid committee]
% \ie{Unnecessary -- remove?}
% A block is due to a \emph{valid committee} if its consensus is reached by the committee defined by its history and it is active at the time of the decision. 
% \end{definition}

We show that if a block is generated by an active committee then its ancestors were approved by active committees. This is enforced for each block by the committee that generated the block, whether it is defined by a reset or by a previous block whose committee is active. 

\begin{lemma} \label{lem:parentActive}
If at time~$t$, a node observes (i.e., receives or generates) a block~$b$ such that block~$b$'s committee is active at~$t$, then a node (perhaps the same node) receives the parent of~$b$,~$\parent{b}$, at time~$t'$, and~$\parent{b}$'s committee is active at~$t'$. 
\end{lemma}

\begin{proof} 
If a node observes a valid block~$b$, approved by the committee defined by it, there are two options: 
If the committee is defined by its parent $\LOneRef{\parent{b}}$, then the committee is active between~$\blockTime{\LOneRef{\parent{b}}}$ and~$t \ge \blockTime{b}$. 
If the committee is defined by a reset pointed by~$\LOneResetRef{b}$, then the committee is active between~$\blockTime{\LOneResetRef{b}}$ and~$t \ge \blockTime{b}$. 
In both cases, a quorum of the committee nodes approved block~$b$ during this interval---recall the block hashes are unpredictable, so the consensus ID is generated on the generation of the primary-ledger block. 
By the lemma assumption, during this time (up to~$t$) the committee is active; therefore, there is a time~$\tilde{t}$ where a correct committee node~$i$ confirmed the correctness of~$\parent{b}$: A node only participates in block generation (lines~\ref{alg:aegis:runConsensusAfterReset} or~\ref{alg:aegis:runConsensus}) after approving its predecessor. 
Node~$i$ logged the parent block~$b'$ and validated its consensus in one of three ways: 
(1) It is the genesis block (line~\ref{alg:aegis:logGenesis}), which needs no validation; 
(2) the parent is checkpointed (line~\ref{alg:aegis:logFromCheckpoint}), either following a series of resets or not; then~$b'$ was approved by an active committee by the time it was checkpointed, as verified by the contract (Algorithm~\ref{alg:contract} line~\ref{alg:contract:timeValidate}); or
(3) the parent is not checkpointed; in this case the node~$i$ that logged~$b$ only logged~$b'$ because the committee of~$b'$ is active at~$\tilde{t}$ (Algorithm~\ref{alg:aegis} line~\ref{alg:aegis:isPrevBlocksCommitteeActive}) and the block is valid (lines~\ref{alg:aegis:chooseBprimeBeforeLoop}, \ref{alg:aegis:chooseBprimeInLoop}). 

In conclusion, the parent block~$b'$ was approved by an active committee. 
\end{proof} 
    
We next deduce that the ancestors of a block approved by an active committee are all approved by active committees. 
The proof follows by induction. 

\begin{lemma} \label{lem:ancestorsValid}
If at time~$t$ a node observes a block~$b$ such that block~$b$'s committee is active at~$t$ then for all ancestors~$b'$ of~$b$, there exist a node (perhaps the same node) and time~$t'$ such that the node receives $b'$ at~$t'$ and the committee of $b'$ is active at~$t'$. 
\end{lemma}

% \begin{proof} 
% We prove by induction on the ancestors starting from the first parent of~$b$ and progressing towards Genesis. 
% As basis, block~$b$'s committee is active by assumption. 
% Next, consider some ancestor~$b'$ of~$b$ and denote its child that is also an ancestor of~$b$ by~$\tilde{b}$. 
% Assuming for the induction that~$\tilde{b}$ is approved by an active committee, it follows by Lemma~\ref{lem:parentActive} that~$b'$ is approved by an active committee, completing the induction. 
% \end{proof} 

%%%%%%%%%%%%%%%%%%%%%%%%%%%%%%%%%%%%%%%%%%%%%%%%%%%%%%%%%%%
%%%%%%%%%%%%%%%%%%%%%%%%%%%%%%%%%%%%%%%%%%%%%%%%%%%%%%%%%%%

        \subsection{Agreement} \label{sec:security:agreement}

For Agreement, we need two more lemmas. 
First, we ensure steps taken by the nodes are reflected in the primary chain. 
We show that if the nodes agree on the ledger up to some height, then a subsequent block will be accepted by the primary-chain contract. 

\begin{lemma} \label{lem:allBlocksCheckpointed}
If there are no two correct nodes that log different blocks at the same height for all indices up to~$k-1$ and 
a node receives a valid block with index~$k$ at time~$t$, then 
a block of height at least~$k$ is checkpointed by step~$t + \deltaActive$ 
and no resets are registered between~$t$ and~$t + \deltaActive$. 
\end{lemma}

Note that the lemma does not state that the registered checkpoint is for a descendant of the valid block at height~$k$ received. 

\begin{proof}
The block's committee is specified by the primary block its parent points to or by a reset on~$\primaryLedger$. 
Denote that $\primaryLedger$ block by~$B$. 
In both cases, some correct node~$i$ in that committee observes block~$b$. 
If within~$\deltaActive$ steps of~$B$ there is no checkpoint, then node~$i$ issues a checkpoint (Algorithm~\ref{alg:aegis} line~\ref{alg:aegis:checkpointLastMinute}). 
We should show the checkpoint is accepted by the contract. 
There cannot be a reset before~$\deltaActive$ has passed, enforced by the contract (line~\ref{alg:contract:resetOnlyAfterStale}). 
If the contract accepts checkpoints for ancestors of~$b$ then~$b$'s checkpoint will still be accepted when it arrives. 
Suppose the contract accepted checkpoints of blocks at height~$k-1$ or earlier that are not ancestors of~$b$. 
The contract only accepts blocks approved by an active committee~(Algorithm~\ref{alg:contract} line~\ref{alg:contract:consensusValidate}). 
So a correct node in a quorum of an active committee approved a block at a height smaller than~$k$, contradicting the lemma assumption. 

Thus, the only case where the checkpoint for~$b$ is not accepted is if a checkpoint for a block at height~$k$ or more was already accepted, completing the proof. 
%
%TODO: Make sure nothing is missing here. 
% If the node observes a new checkpoint~$B'$ for an ancestor of~$b$ it will not necessarily issue a checkpoint for~$b$. 
% However, the block pointed-to by~$B'$, denoted~$b'$, specifies a new committee, and this committee will issue a checkpoint for the latest checkpoint observed by the creators of~$b'$. 
% This holds recursively ensuring a checkpoint is written at least every step and at most every~$\deltaActive$ steps. 
% Each checkpoint points to a subsequent~$\ledger$ block, eventually reaching~$b$. 
\end{proof}

We also show that consensus agreement applies to decisions taken by active committees. 

\begin{lemma}[Active committee agreement] \label{lem:activeCommitteeAgreement} 
In an execution of Aegis, an active committee does not decide different values for the same consensus instance. 
\end{lemma} 

\begin{proof}
Consider an execution~$\sigma$ of Aegis where two correct nodes in a committee active in the time range~${[t_0, t_f]}$ decide in steps~${t_0 \le t_1, t_2, \le t_f}$ two values~$v_1$ and~$v_2$ (one each) for the same consensus instance. 

Let~$\sigma'$ be an execution with the same prefix as Aegis up to~$t_1$, and extended such that the committee remains active forever (its nodes never unstake), as in the consensus protocol assumptions. 
By the agreement property of the consensus protocol, the two nodes cannot decide distinct values. 

Since at~$t_1$ the nodes cannot distinguish between~$\sigma$ and~$\sigma'$, the nodes take the same steps in~$\sigma$ up to~$t_1$, so~${v_1 = v_2}$. 
\end{proof}

We are now ready to prove agreement. 

\begin{proposition}[Agreement] \label{prp:agreement}
Two correct nodes do not log different blocks at the same height. 
\end{proposition}

\begin{proof} 
We prove by induction on the block number~$k$. 

\textbf{Base.  } 
For~$k = 0$ all correct nodes at all times log the Genesis block (Algorithm~\ref{alg:aegis} line~\ref{alg:aegis:logGenesis}; $\ledger(0)$ is never updated elsewhere). 
The contract initially has no checkpoints. 

\textbf{Assumption.  } 
For~$k > 0$, assume that the claim holds for~${k-1}$. 
That is, for all nodes~$i$ and~$j$ (maybe~$i=j$) and times~$t_1 \le t_2$, if node~$i$ logs a block~$b$ at height~$k' \le k-1$ at time~$t_1$ ($\ledger_i^{t_1}(k') \neq \bot$) and node~$j$ logs block~$b'$ at height~$k'$ at time~$t_2$ ($\ledger_j^{t_2}(k') \neq \bot$), then~$\ledger_i^{t_1}(k') = \ledger_j^{t_2}(k')$. 
Also assume there is no checkpoint for a block~$b$ at a height~$k' \le k-1$ and a node~$i$ that logs a block~$b' \neq b$ at height~$k'$. 

\textbf{Step.  } 
Now we prove for~$k$. 
Assume for contradiction that there exist nodes~$i$ and~$j$ (maybe~${i=j}$) and times~$t_1 \le t_2$ such that node~$i$ logs a block~$b$ in position~$k$ at time~$t_1$ and a node~$j$ logs block $b' \neq b$ in position~$k$ at time~$t_2$: $\ledger_i^{t_1}(k) = b \neq b' = \ledger_j^{t_2}(k)$. 

The content of the ledger in positive positions is determined due to checkpoints (line~\ref{alg:aegis:logFromCheckpoint}) or consensus (lines~\ref{alg:aegis:logByPrevConsensus} and~\ref{alg:aegis:logByCurrConsensus}). 

% \IE{The value of $t_0$ is checked in line 32 but then updated in line 39. The later update is checked (for checkpointing) in line 41. Is this what should happen? It means that the node keeps processing Aegis blocks even after the previous checkpoint was deprecated, counting on the fact someone checkpoints everything as necessary and there aren't any erroneous resets during that time.} 

% We observe that if a node logs a block, that block was generated by a committee that was active at the time of the decision.
If a node logs a block~$b$ in position~$k$ due to consensus (lines~\ref{alg:aegis:logByPrevConsensus} or~\ref{alg:aegis:logByCurrConsensus}), then it is because 
it received active committee consensus message for~$b$ (activity period verified in line~\ref{alg:aegis:isPrevBlocksCommitteeActive}). 
The committee is defined by~$b$'s parent or specified by a reset. 

If a node logs a block~$b$ in position~$k$ due to a checkpoint (line~\ref{alg:aegis:logFromCheckpoint}), 
then it is because~$b$ is an ancestor of a checkpointed block~$\hat{b}$. 
The checkpointed block~$\hat{b}$ is due to an active committee verified by the contract (Algorithm~\ref{alg:contract} line~\ref{alg:contract:consensusValidate}). 
Therefore~(Lemma~\ref{lem:ancestorsValid}), all ancestors of~$\hat{b}$ are also approved by active committees. 
In particular, block~$b$ is approved by an active committee. 
That committee is either the one defined by the previous block, or specified by a~reset. 

In summary, both blocks~$b$ and~$b'$ are created by active committees, each either defined by a reset or by the committee referenced by the previous block. 

We now consider each of the possible cases. 
By the induction assumption, both~$i$ and~$j$ agree on the previous block. 
If the committees of both blocks are specified by the previous Aegis block, since by assumption they agree on this prior block, that committee decided conflicting values, contradicting the consensus protocol agreement property (Lemma~\ref{lem:activeCommitteeAgreement}). 

Next, if the committees of both~$b$ and~$b'$ are specified by reset blocks, we show it is the same reset block. 
Assume for contradiction that block~$b$ (without loss of generality) is confirmed by the first reset and that~$b'$ is confirmed by a later reset. 
Since block~$b$ is approved by the first reset committee (reset at time~$t_\textit{reset}$), 
by Lemma~\ref{lem:allBlocksCheckpointed} a block of height at least~$k$ is issued a checkpoint before the reset committee becomes stale, with the checkpoint being written by time~$t_\textit{reset} + \deltaActive$ for a block of height at least~$k$.
The contract will only accept another reset after~$t_\textit{reset} + \deltaActive$ (Algorithm~\ref{alg:contract} line~\ref{alg:contract:resetOnlyAfterStale}), so after the checkpoint. 
A committee following a reset considers the latest checkpoint (Algorithm~\ref{alg:aegis} line~\ref{alg:aegis:checkpointBeforeReset}), so a block~$b'$ confirmed by the committee specified by the second reset cannot extend the chain earlier than~$k+1$ and cannot have the same height as~$b$. 
A contradiction. 

Finally, consider the case where the committee of~$b$ (without loss of generality) is specified by its parent and that of~$b'$ is specified by a reset. 
There are two possible cases. 
In the first case, $\LOneRef{b_{k-1}}$ is before~$\LOneResetRef{b'}$ in the primary chain. 
All nodes agree on all blocks up to~$k-1$ and an active node observed block~$b$ with height~$k$, therefore (Lemma~\ref{lem:allBlocksCheckpointed}) a block at height~$k$ or more is checkpointed before the contract accepts any reset. 
So a block at height~$k$ is checkpointed before the reset defining the committee of block~$b'$. 
An Aegis block~$b'$ following such a subsequent reset cannot result in a block at height~$k$, a contradiction. 

The alternative case is that~$\LOneResetRef{b'}$ is before $\LOneRef{b_{k-1}}$. 
In this case block~$b'$ is not valid due to the order of its parent reference and reset (Algorithm~\ref{alg:isValid} line~\ref{alg:isValid:resetAndParentRefOrder}), so no node would log it. Again, a contradiction. 

Having reviewed all cases and reached contradictions, we conclude that no two nodes log different blocks at the same height. 
Since the contract will only checkpoint a block at height~$k$ if its committee is active, a checkpoint can only be produced if an active committee decides a different block at height~$k$. 
We have shown this is impossible, thus a checkpointed block is not different from a valid block received by a node at height~$k$. 

This completes the induction step and thus the proof. 
\end{proof}

%%%%%%%%%%%%%%%%%%%%%%%%%%%%%%%%%%%%%%%%%%%%%%%%%%%%%%%%%%%
%%%%%%%%%%%%%%%%%%%%%%%%%%%%%%%%%%%%%%%%%%%%%%%%%%%%%%%%%%%

        \subsection{No-Batching Validity} \label{sec:security:validity}

To prove no-batching validity, we first note that the consensus protocol guarantees validity of an active committee decision. 
The proof is similar to that of Lemma~\ref{lem:activeCommitteeAgreement}. 

\begin{lemma}[Active committee validity] \label{lem:activeCommitteeValidity} 
In an execution of Aegis, if all committee members of an active committee have the same input value, and Aegis executes the consensus protocol of this committee, then no correct node outputs a different value. 
\end{lemma} 

We are now ready to prove no-batching validity. 

\begin{proposition}[No-batching validity] \label{prp:validity}
Aegis achieves no-batching validity. 
\end{proposition}
    
\begin{proof}
We prove by induction on the block number~$k$. 

\textbf{Base.  } 
For~$k = 0$ all correct nodes at all times log the Genesis block (Algorithm~\ref{alg:aegis} line~\ref{alg:aegis:logGenesis}; $\ledger(0)$ is never updated elsewhere). 
The contract initially has no checkpoints. 

\textbf{Assumption.  } 
For~$k > 0$, assume that the claim holds for~${k-1}$. 
% That is, for all nodes~$i$, if node~$i$ logs a block~$b$ at height~$k' \le k-1$ at time~$t_1$ ($\ledger_i^{t_1}(k') \neq \bot$) and all nodes~$j$ are correct with input~$v_{k'}$ for height~$k'$, then~${\ledger_i^{t_1}(k') = v_{k'}}$. 
That is, for all nodes~$i, i'$ (perhaps~$i = i'$), if node~$i'$ logs a block~$b$ at height~$k' \le k-1$ and node~$i$ logs a block at height~$k' - 1$ at time~$t$ and all nodes are correct with input~$v$ from time~$t$ ($\forall j, t'' > t: \inputFuncITK{j}{t''}{k'} = v_{k'}$) then node~$i'$ logs~$v$ at height~$k'$:~${\ledger_{i'}^{t_1}(k') = v_{k'}}$. 
Also assume there is no checkpoint for a block~$b$ at a height~$k' \le k-1$ with value different from~$v_{k'}$. 

\textbf{Step.  } 
Now we prove for~$k$. 
Consider 
% Assume for contradiction that there exist 
a node~$i$ and time~$t$ such that node~$i$ logs a block~$b$ in position~$k$ at time~$t$ and $\ledger_i^{t_1}(k) = b \neq v_{k}$. 

The content of the ledger in positive positions is determined due to checkpoints (line~\ref{alg:aegis:logFromCheckpoint}) or consensus (lines~\ref{alg:aegis:logByPrevConsensus} and~\ref{alg:aegis:logByCurrConsensus}). 

If a node logs a block~$b$ in position~$k$ due to consensus (lines~\ref{alg:aegis:logByPrevConsensus} or~\ref{alg:aegis:logByCurrConsensus}), then it is because it received active committee consensus message for~$b$ (activity period verified in line~\ref{alg:aegis:isPrevBlocksCommitteeActive}). 
The committee is defined by~$b$'s parent or specified by a reset. 

If a node logs a block~$b$ in position~$k$ due to a checkpoint (line~\ref{alg:aegis:logFromCheckpoint}), 
then it is because~$b$ is an ancestor of a checkpointed block~$\hat{b}$. 
The checkpointed block~$\hat{b}$ is due to an active committee verified by the contract (Algorithm~\ref{alg:contract} line~\ref{alg:contract:consensusValidate}). 
Therefore~(Lemma~\ref{lem:ancestorsValid}), all ancestors of~$\hat{b}$ are also approved by active committees. 
In particular, block~$b$ is approved by an active committee. 
That committee is either the one defined by the previous block, or specified by a reset. 

In both cases, the active committee nodes call the input function only after observing the~${k - 1}$ block (lines~\ref{alg:aegis:runConsensus} or~\ref{alg:aegis:runConsensusAfterReset}). 
Thus, in both cases, all inputs of the consensus nodes are read from the input function $\inputFuncITK{}{}{\cdot}$ after they logged block~${k - 1}$. 
If the input function returns the same value~$v_k$ for all correct nodes after the first decision on block~${k - 1}$, this will be the input of all the active-committee nodes for block~$k$. 
The consensus validity property (Lemma~\ref{lem:activeCommitteeValidity}) guarantees that the consensus protocol outputs this value and this is their decision. 

This completes the induction step and thus the proof. 
\end{proof}

%%%%%%%%%%%%%%%%%%%%%%%%%%%%%%%%%%%%%%%%%%%%%%%%%%%%%%%%%%%
%%%%%%%%%%%%%%%%%%%%%%%%%%%%%%%%%%%%%%%%%%%%%%%%%%%%%%%%%%%

        \subsection{Progress} \label{sec:security:progress}

Finally, we show Aegis extends the ledger after~$\gst$. 
To prove progress, we first show that consensus termination applies to active Aegis committees after~$\gst$. 

\begin{lemma}[Active committee termination] \label{lem:activeCommitteeTermination} 
In an execution of Aegis, after~$\gst$, a node in an active committee decides within~$\deltaConsensus$ steps. 
\end{lemma} 

\begin{proof}
Consider an execution~$\sigma$ of Aegis where a node in an active committee participates in a consensus instance in a step~${t \ge \gst}$. 

Let~$\sigma'$ be an execution with the same prefix as Aegis up to~${t + \deltaConsensus}$, and extended such that the committee remains active forever (its nodes never unstake), as in the consensus protocol assumptions. 
By the termination property of the consensus protocol, the node decides by~$t + \deltaConsensus$. 

Since at~$t + \deltaConsensus$ the nodes cannot distinguish between~$\sigma$ and~$\sigma'$, the nodes take the same steps in~$\sigma$ up to~$t + \deltaConsensus$, so the node decides in~$\sigma$ by~$t + \deltaConsensus$ in~$\sigma$. 
\end{proof}

We can now prove that Aegis achieves the progress. 

\begin{proposition}[Progress] \label{prp:progress} 
Aegis guarantees progress. 
\end{proposition}

% For all executions, there exists a time~$t > \gst$ (maybe a different~$t$ for different executions) such that for all~$t' > t$, if the latest decided value by any node at~$t'$ is~$k$, then by~$t' + \deltaConsensus + \deltaPropagation$ all nodes decide on~${k + 1}$

\begin{proof} 
A node adds a block to the Aegis ledger either since its descendant is checkpointed (Algorithm~\ref{alg:aegis} line~\ref{alg:aegis:logFromCheckpoint}) or since it is approved by an active committee (Algorithm~\ref{alg:aegis} line~\ref{alg:aegis:logByPrevConsensus} or~\ref{alg:aegis:logByCurrConsensus}). 
In both cases, the result is that within at most~$\deltaActive$ steps it is checkpointed in the primary chain and seen by all correct nodes (Algorithm~\ref{alg:aegis} line~\ref{alg:aegis:checkpointLastMinute}).

Consider the following worst case scenario.
At time~$t > \gst$ the Aegis checkpoint is later than~$t - \deltaActive + 3\deltaPrimaryWrite$, the nodes take consensus steps and decide within~$\deltaConsensus$ steps (Algorithm~\ref{alg:aegis} line~\ref{alg:aegis:runConsensus}, Lemma~\ref{lem:activeCommitteeTermination}). 
However, if the consensus runs for too long, and its result might not be checkpointed in time, it is abandoned (Algorithm~\ref{alg:aegis} line~\ref{alg:aegis:giveUpTooLate}). 
Then no progress can be made until~$\deltaActive$ after the last checkpoint/reset and a new committee is instated with a new reset (line~\ref{alg:aegis:resetAfterCheckpoint}).
The reset is registered and nodes see it within~$2\deltaPrimaryWrite$ steps, and a new committee reaches consensus within~$\deltaConsensus$ (Lemma~\ref{lem:activeCommitteeTermination}).

Other scenarios only speed up progress:
If the checkpoint is earlier, the reset will be issued faster, and if the consensus following the checkpoint succeeds the block will be created faster.

Overall, the next block is created at most~$\progressStabilizationTime = \deltaConsensus + \deltaPropagation + \deltaActive + 2\deltaPrimaryWrite$ steps after~$\gst$.

From that point on, a new consensus is reached within~$\deltaConsensus$ steps (Lemma~\ref{lem:activeCommitteeTermination}) and propagated to all correct nodes within~\deltaPropagation, at which point a new consensus instance is started. 
The interval between blocks is thus bounded by~${\deltaConsensus + \deltaPropagation}$, as required for progress. 
\end{proof} 

%%%%%%%%%%%%%%%%%%%%%%%%%%%%%%%%%%%%%%%%%%%%%%%%%%%%%%%%%%%
%%%%%%%%%%%%%%%%%%%%%%%%%%%%%%%%%%%%%%%%%%%%%%%%%%%%%%%%%%%

        \subsection{Slashing} \label{sec:security:slashing}

Enforcing the desired behavior of correct nodes while they are active necessitates penalizing them if they successfully violate safety (Agreement or Validity). 
We show Aegis achieves this if the time it takes to unstake is long enough.

\begin{proposition}[Penalty] \label{prp:penalty} 
Assume~$\deltaActive > 3\deltaPrimaryWrite$.
If a node~$i$ belongs to an active committee that violates the consensus safety, it will be penalized. 
\end{proposition} 

\begin{proof}
With even a single Byzantine node in the committee, validity vacuously holds. 
Thus, the only potential safety violation is for agreement, resulting in two distinct nodes deciding different decisions. 
If a node decides a value at height~$k$, then this or another node will issue a checkpoint for a subsequent block at height at least~$k$ within~$\deltaActive - 3\deltaPrimaryWrite$ steps (Algorithm~\ref{alg:aegis} line~\ref{alg:aegis:checkpointLastMinute}). 
Thus, if a node~$i$ decides a value~$v_i^k$ and a node~$j$ decides a value~$v_j^k \neq v_i^k$, then one of them will find out once the first checkpoint is written to the primary chain or earlier through direct communication of a consensus step (line~\ref{alg:aegis:runConsensusAfterReset} or line \ref{alg:aegis:runConsensus}) or from received blocks (line~\ref{alg:aegis:logByPrevConsensus}). 
Once a node detects such an agreement violation, it issues a forensics procedure by publishing the necessary information to the blockchain. 
This process is completed within~$\deltaPrimaryWrite$ steps, at which point its counterpart does the same, completing the on-chain forensics process while the committee of the violating nodes is still active. 
Thus, the contract can penalize the stake, which is yet to be withdrawn. 
\end{proof}
    
%%%%%%%%%%%%%%%%%%%%%%%%%%%%%%%%%%%%%%%%%%%%%%%%%%%%%%%%%%%
%%%%%%%%%%%%%%%%%%%%%%%%%%%%%%%%%%%%%%%%%%%%%%%%%%%%%%%%%%%

        \subsection{Aegis Security} \label{sec:security:conclusion} 

The conclusion follows directly from Propositions~\ref{prp:agreement},~\ref{prp:validity}, and~\ref{prp:progress}, since nodes are penalized (Proposition~\ref{prp:penalty}) for violating safety. 

\begin{theorem} \label{thm:correctness}
Aegis ensures Agreement, No-Batching Validity, and Progress. 
\end{theorem}

%%%%%%%%%%%%%%%%%%%%%%%%%%%%%%%%%%%%%%%%%%%%%%%%%
%%%%%%%%%%%%%%%%%%%%%%%%%%%%%%%%%%%%%%%%%%%%%%%%%
%%%%%%%%%%%%%%%%%%%%%%%%%%%%%%%%%%%%%%%%%%%%%%%%%

    \section{Practical Considerations and Broader Applicability} \label{sec:practicalConcerns} 

% Practical concerns
While Aegis can be deployed as-is on top of a primary blockchain, there are two practical aspects to consider. 
Aegis requires primary-chain transactions in intervals close to the time for unstaking~$\deltaActive$. 
This implies a tradeoff between short unstaking time and low primary-chain overhead.

In terms of performance, Aegis can change committee membership at a lower frequency than every block. 
This allows using contemporary BFT state machine replication algorithms (e.g.,~\cite{yin2019hotstuff,gelashvili2022jolteon,babel2023mysticeti}) that are more efficient than running a one-shot consensus instance per block.

% General epoch changes
Beyond stake-based systems, our tethering approach provides a principled solution for secure epoch changes in blockchains more broadly.
Aegis can be adapted to support various committee transition criteria instead of staking and unstaking.
For example, it could trigger based on a certain number of blocks (time-based epochs), based on consensus of a subcommittee, or by a central authority.
To implement such extensions, the smart contract on the primary chain would need additional logic to verify the transition conditions before accepting resets.
% However, this alone is insufficient for preventing forks.
% A critical requirement is that enough correct nodes must be able to independently determine whether an epoch change is occurring without reading the state of the primary chain.
% This ensures that the old committee knows exactly when to stop producing blocks, preventing potential forks during transitions.
In addition, this would require tailoring the node protocol to match the transition logic in the smart contract in order to maintain the security guarantees of Aegis.

%%%%%%%%%%%%%%%%%%%%%%%%%%%%%%%%%%%%%%%%%%%%%%%%%
%%%%%%%%%%%%%%%%%%%%%%%%%%%%%%%%%%%%%%%%%%%%%%%%%
%%%%%%%%%%%%%%%%%%%%%%%%%%%%%%%%%%%%%%%%%%%%%%%%%

    \section{Conclusion} \label{sec:conclusion} 

% Summary 
We present Aegis, a novel tethered blockchain protocol whose nodes can join and leave in a permissionless fashion by staking through a primary-chain smart contract.
Aegis guarantees safety at all times and progress starting shortly after $\gst$ despite asynchronous communication among its nodes up to~$\gst$.
A core assumption allowing Aegis to succeed is that the nodes have synchronous access to the primary chain. 
Without this assumption, no protocol can achieve all desired properties: Agreement, No-Batching Validity and Progress.

The tools to deploy staking for Aegis are ready at hand, e.g., Avalanche subnets~\cite{avalanche2022subnets}; direct implementation with a smart contract; and Ethereum's EigenLayer~\cite{eigenlayer}, which enables a market for stake on Ethereum that secures external tasks as well as Ethereum's consensus.
Thus, Aegis opens the door for a new generation of tethered blockchains supporting truly decentralized applications.

%%%%%%%%%%%%%%%%%%%%%%%%%%%%%%%%%%%%%%%%%%%%%%%%%
%%%%%%%%%%%%%%%%%%%%%%%%%%%%%%%%%%%%%%%%%%%%%%%%%
%%%%%%%%%%%%%%%%%%%%%%%%%%%%%%%%%%%%%%%%%%%%%%%%%

\bibliographystyle{ACM-Reference-Format}
\bibliography{btc}

%%%%%%%%%%%%%%%%%%%%%%%%%%%%%%%%%%%%%%%%%%%%%%%%%
%%%%%%%%%%%%%%%%%%%%%%%%%%%%%%%%%%%%%%%%%%%%%%%%%
%%%%%%%%%%%%%%%%%%%%%%%%%%%%%%%%%%%%%%%%%%%%%%%%%

\appendix

    \section{Notation} \label{app:notation} 

% Notation table 
\begin{table}[t!bp]
\centering
\renewcommand{\arraystretch}{1.4}
\caption{Notation}
\label{tab:notation}
\begin{tabular}{|l|p{0.7\columnwidth}|}
\hline
\textbf{Symbol} & \textbf{Meaning} \\ \hline
% $N$ & Number of nodes \\ \hline
\consensusThreshold & Minimum ratio of correct nodes in an active committee \\ \hline
\gst & Global Stabilization Time \\ \hline
\deltaPropagation & Network propagation bound after~$\gst$ \\ \hline
\deltaActive & Time between unstake order and funds withdrawal \\ \hline
\deltaPrimaryWrite & Time to write a block to the primary chain \\ \hline
\deltaConsensus & Time to reach consensus after~$\gst$ \\ \hline
$\hash(B) = \hat{B}$ & Hash of block~$B$ \\ \hline
\primaryLedger & Primary ledger \\ \hline
% $\primaryLedger(k)$ & Primary block at height~$k$ \\ \hline
$\primaryLedger(\hat{B})$ & Block~$B$ if~$B \in \primaryLedger$, otherwise~$\bot$ \\ \hline
\ledger & Aegis ledger \\ \hline
$\ledger(k)$ & Aegis block at height~$k$ \\ \hline
$\blockParent{b}$ & Parent block of block~$b$ \\ \hline
$\mathcal{B}(\hat{b})$ & Block~$b$ if~$b$ is in block set $B$, otherwise~$\bot$ \\ \hline
$\LOneRef{b}$ & Reference from block~$b$ to a primary block \\ \hline
$\LOneResetRef{b}$ & Reference from block~$b$ to a reset in a primary block \\ \hline
$\LTwoRef{B}$ & Checkpoint reference from primary block~$b$ to an Aegis block \\ \hline
$\blockMembership{B}$ & Committee defined by primary block~$B$ \\ \hline
$\blockTime{b}$ & Generation time of block~$b$ \\ \hline
\end{tabular}
\end{table}

Table~\ref{tab:notation} summarizes key notation used in this paper.

%%%%%%%%%%%%%%%%%%%%%%%%%%%%%%%%%%%%%%%%%%%%%%%%%%%%%%%%%%%
%%%%%%%%%%%%%%%%%%%%%%%%%%%%%%%%%%%%%%%%%%%%%%%%%%%%%%%%%%%

\section{Security Assumptions} \label{app:securityAssumptions}

To provide a comprehensive view of Aegis's security model, we reiterate its security assumptions.

\begin{enumerate}
    \item \textbf{Primary Blockchain Properties}~(\secQuickRef{sec:model:elements})
    \begin{enumerate}
        \item \textit{Primary blockchain security}. The primary chain is implemented by a trusted party that maintains token balances and staked amounts and supports smart contracts.
        \item \textit{Unpredictable block hashes}. Nodes cannot guess future primary chain block hashes except with negligible probability.
    \end{enumerate}

    \item \textbf{Node Behavior and Stake Properties}~(\secQuickRef{sec:model:elements})
    \begin{enumerate}
        \item \textit{Sufficient unstaking delay}. Unstaking completes~$\deltaActive$ steps after the request transaction is placed.
        \item \textit{Honest majority in committees}. In all committees, the ratio of stake held by correct nodes exceeds~$\consensusThreshold$.
        \item \textit{Stake-based correctness incentive}. Correct nodes follow the protocol while staked to avoid losing stake.
    \end{enumerate}

    \item \textbf{Network Assumptions}~(\secQuickRef{sec:model:elements})
    \begin{enumerate}
        \item \textit{Synchronous primary chain access}. Nodes can observe state and issue transactions within bounded time~$\deltaPrimaryWrite$.
        \item \textit{Partially synchronous node communication}. After~$\gst$, message delivery between nodes is bounded by~$\deltaPropagation$.
        \item \textit{Block propagation persistence}. Blocks are delivered to all nodes, even those staking later.
    \end{enumerate}

    \item \textbf{Consensus Algorithm Properties}~(\secQuickRef{sec:model:consensus})
    \begin{enumerate}
        \item \textit{Underlying consensus guarantees}. If the ratio of correct nodes in a committee exceeds~$\consensusThreshold$, the consensus protocol ensures agreement (no two correct nodes decide different values), validity (if all nodes are correct with same input, no other value is decided), and termination (after~$\gst$, all correct nodes decide within~$\deltaConsensus$ time).
        \item \textit{Consensus validation function}. There exists a function~$\consensusValidate$ that verifies block validity.
        \item \textit{Forensics support}. The consensus protocol allows identifying Byzantine nodes.
    \end{enumerate}

    \item \textbf{Timing Constraints}~(\secQuickRef{sec:model:consensus})
    \begin{enumerate}
        \item \textit{Minimum unstaking delay}. Unstaking is slow enough compared to the time to issue a transaction to the primary chain:~${\deltaActive > 3\deltaPrimaryWrite}$.
    \end{enumerate}
\end{enumerate}

\end{document}